\RequirePackage[l2tabu, orthodox]{nag}		% checks for outdated latex commands ... you can remove this if you want
\documentclass[11pt]{article}

\usepackage[T1]{fontenc}
\usepackage{lmodern}
\usepackage[frenchmath]{mathastext}		% sets default math font to the default text font (in italics)														% improves spacing
\usepackage{natbib}       									
\usepackage{amsmath}                        	% need for subequations
\numberwithin{equation}{section}
\usepackage{graphicx} 				% if you compile LaTeX => PS => PDF, then uncomment this and uncomment next two lines
\usepackage{adjustbox,float,wrapfig}
\usepackage[font=footnotesize]{caption}
\graphicspath{{Figures/}}
\usepackage{enumitem}									% allows you to customize spacing of lists
\usepackage{mdwlist}									% enumerate* itemize* are compact
\usepackage[dvipsnames]{xcolor}							% gives you the colors you will use in the next line
\usepackage[plainpages=false, pdfpagelabels]{hyperref} 			% enables and customizes hyperlinks
	\hypersetup{
		colorlinks   = true,
		citecolor    = RoyalBlue, 			%NavyBlue,
		linkcolor    = RubineRed, 			%Maroon,
		urlcolor     = Turquoise
	}
\usepackage{amssymb}                        				% gives you \mathbb{} font
\usepackage[mathscr]{eucal}                 				% gives you \mathscr font
\usepackage{dsfont}							% gives you \mathds{} font

\usepackage{mathtools}                      % need for `show only references'
\mathtoolsset{showonlyrefs=true}          % only equations which are labeled AND referenced will numbered use \eqref{} instead of (\ref{})
\usepackage{amsthm}                         % need for theorem-proof environment
\allowdisplaybreaks                       % allows page breaks for long equations...you can prevent a page-break with \\*
\theoremstyle{plain}
\newtheorem{theorem}{Theorem}
\numberwithin{theorem}{section}
\newtheorem{lemma}[theorem]{Lemma}       	% [theorem] ==> theorems and lemmas will share a counter
\newtheorem{proposition}[theorem]{Proposition}
\newtheorem{corollary}[theorem]{Corollary}
\theoremstyle{definition}

\newtheorem{remark}[theorem]{Remark}
\newtheorem*{remark*}{Remark}

\newcommand{\E}{\mathbb{E}}

\renewcommand{\(}{\left(}
\renewcommand{\)}{\right)}
\renewcommand{\[}{\left[}
\renewcommand{\]}{\right]}

\newcommand\Cb{\mathds{C}}
\newcommand\Eb{\mathds{E}}

\newcommand\Pb{\mathds{P}}

\newcommand\Rb{\mathds{R}}

\newcommand\Dc{\mathscr{D}}

\newcommand \al {\alpha}
\newcommand\eps{\varepsilon}

\newcommand\sig{\sigma}

\newcommand\del{\delta}

\newcommand\kap{\kappa}
\newcommand \cb {c_0}

\newcommand \pib {\pi_0}

\newcommand\Fv{\textbf{F}} %<--- with concrete fonts there is no bolded math font so you changed to bold text font
\newcommand\Gv{\textbf{G}} %<--- with concrete fonts there is no bolded math font so you changed to bold text font
\newcommand\xv{\mathbf{x}}
\newcommand\yv{\mathbf{y}}

\newcommand\pih{\hat{\pi}}

\newcommand \Uh {\widehat{U}}

\newcommand\ch{\hat{c}}

\newcommand\tauh{\hat{\tau}}

\newcommand\Vt{\widetilde{V}}

\newcommand\dd{\mathrm{d}}
\newcommand\ee{\mathrm{e}}

\providecommand{\keywords}[1]{\textbf{\textit{Keywords: }} #1}

\newcommand \ws {w^*}
\newcommand \wa {w_\al}
\newcommand \wo {w_1}
\newcommand \ys {y^*}
\newcommand \etas {\eta^*}
\newcommand \pis {\pi^*}
\newcommand \cs {c^*}

\newif\ifcomment

\commentfalse
\ifcomment
	\usepackage[paperwidth=8.5in,paperheight=11in,top=1.25in, bottom=1.25in, left=0.6in, right=1.4in, marginparwidth=1.2in]{geometry}
	\newcommand{\mcomment}[1]{\marginpar{\tiny\textbf{\color{blue} #1}}}
\else
	\usepackage[paperwidth=8.5in,paperheight=11in,top=1.25in, bottom=1.25in, left=1in, right=1in]{geometry} %use with 10pt font
	\newcommand{\mcomment}[1]{}
\fi

\begin{document}

\title{Optimal Dividend Distribution Under Drawdown and Ratcheting Constraints on Dividend Rates}

\author{
Bahman Angoshtari
\thanks{Department of Applied Mathematics, University of Washington.  \textbf{e-mail}: \url{bahmang@uw.edu}}
\and
Erhan Bayraktar
\thanks{Department of Mathematics, University of Michigan.  \textbf{e-mail}: \url{erhan@umich.edu} E. Bayraktar is supported in part by the National Science Foundation under grant DMS-1613170 and by the Susan M. Smith Professorship.}
\and
Virginia R. Young
\thanks{Department of Mathematics, University of Michigan.  \textbf{e-mail}: \url{vryoung@umich.edu} V. R. Young is supported in part by the Cecil J. and Ethel M. Nesbitt Professorship.}
}

\date{This version: \today}

\maketitle

\begin{abstract}
We consider the optimal dividend problem under a habit-formation constraint that prevents the dividend rate to fall below a certain proportion of its historical maximum, a so-called {\it drawdown} constraint.  Our problem is an extension of Duesenberry's optimal-consumption problem under a ratcheting constraint, studied by \cite{Dybvig1995}, in which consumption is restrained to be nondecreasing.  Our problem also differs from Dybvig's in that the time of ruin could be finite in our setting, whereas ruin is impossible in Dybvig's work.  We formulate our problem as a stochastic control problem with the objective of maximizing the expected discounted utility of the dividend stream until bankruptcy, in which risk preferences are embodied by power utility.  We write the corresponding Hamilton-Jacobi-Bellman variational inequality as a nonlinear, free-boundary problem and solve it semi-explicitly via the Legendre transform.  The optimal (excess) dividend rate $c^*_t$ - as a function of the company's current surplus $X_t$ and its historical running maximum of the (excess) dividend rate $z_t$ - is as follows: There are constants $0 < \wa < \wo < \ws$ such that (1) for $0 < X_t \le \wa z_t$, it is optimal to pay dividends at the lowest rate $\al z_t$, (2) for $\wa z_t < X_t < \wo z_t$, it is optimal to distribute dividends at an intermediate rate $c^*_t \in (\al z_t, z_t)$, (3) for $\wo z_t < X_t < \ws z_t$, it is optimal to distribute dividends at the historical peak rate $z_t$, (4) for $X_t > \ws z_t$, it is optimal to increase the dividend rate above $z_t$, and (5) it is optimal to increase $z_t$ via singular control as needed to keep $X_t \le \ws z_t$.  Because, the maximum (excess) dividend rate will eventually be proportional to the running maximum of the surplus, ``mountains will have to move'' before we increase the dividend rate beyond its historical maximum.
\end{abstract}


%-----------------------------------------------------------------------------------
%
%       SECTION: 		Introduction
%
%-----------------------------------------------------------------------------------

\section{Introduction}

One of the fundamental goals of risk managers is to improve the stability of companies that operate in risky environments. This goal can be reached through the choice of the dividend policy, that is, how much of a company's surplus to pay out (or, equivalently, retain). There is a tradeoff embedded in this decision. Paying out more dividends would increase a company's worth in the eyes of its shareholders, while doing so would also reduce future reserves that are essential for the survival of the company through financial hardships.  Since the seminal work of \cite{DeFinetti1957}, there has been an active area of research devoted to finding optimal payment of dividends under various criteria.

Shareholders and analysts react negatively (and arguably overreact) when the rate of dividend payment decreases. The existing literature on optimal dividend policies, however, mainly ignores shareholder's averseness to decreases in dividend payments. Our main goal in this paper is to address this issue by considering a so-called \emph{drawdown constraint} on the rate of dividend payments, that is, we demand that the future rate of dividend payments cannot go below a fixed proportion of the historical maximum of the dividend rate to date.  Interestingly, we find that the running maximum of the optimal dividend rate is (eventually) proportional to the running maximum of the surplus process. In other words, ``mountains will have to move'' before increasing the dividend rate beyond it historical peak.

% \mcomment{Erhan: we should in the first two paragraphs tell that surprisingly it turns out that the optimal running maximum dividend rate turns out to be proportional to the running maximum of the surplus process X. Which tells the story that the mountains need to move before we increase the dividend rate. In fact, this could be tied with Albrecher et al., since they actually take their strategy of this form (function of the running maximum), we show that the optimal strategy turns out to be a function of the maximum wealth. (Their functions are too trivial as you said, because they allow for one time change only.)\vspace{1em}}

Many of the existing results in optimal dividends show that to maximize the expected discounted dividends until bankruptcy, it is optimal to pay dividends according to a band strategy or its special case of a barrier strategy; see \cite{Avanzi2009} for a survey.  For example, \cite{AsmussenTaksar1997} consider the optimal dividend strategy for an insurance company whose surplus process follows Brownian motion with drift.  They assume that the insurer pays dividends to maximize the expectation of discounted dividends paid between now and ruin, and they consider two cases. In the first case, \cite{AsmussenTaksar1997} constrain the dividend rate to lie in an interval $[0, C_0]$.  If $C_0$ is less than or equal to a critical value $C^*$, then it is optimal to pay dividends at the rate $C_0$ at all levels of the surplus.  If $C_0 > C^*$, then it is optimal to pay dividends at rate $0$ if surplus is lower than some value $X^*$ and at rate $C_0$ if surplus is greater than $X^*$.  In the second case, \cite{AsmussenTaksar1997} do not restrict the rate of dividend payments, and it is optimal to pay dividends via a barrier strategy, under which all the surplus in excess of a barrier $b$ are paid out as dividends.  Also, see \cite{AsmussenHojgaardTaksar2000} for an extension of \cite{AsmussenTaksar1997} in which the authors allow the insurer to control its surplus via reinsurance.  Finally, see  \cite{GerberShiu2006} for explicit calculations related to \cite{AsmussenTaksar1997}.

Note that paying dividends at rate $C_0$ if surplus is greater than $X^*$ and at rate $0$ if surplus is less than $X^*$ results in a volatile all-or-nothing path.  Incorporating habit formation is a way to smooth the rate of dividend payments; see \cite{Constantinides1990} for seminal work in habit formation.  We model habit formation via a drawdown constraint on the rate of excess dividend payments; by ``excess'' we mean in \textit{excess} of the interest paid if the surplus were invested completely at the riskless rate of return.  We require that the rate of excess dividend payments never falls below some given fraction of the historical maximum rate of excess dividend payments.  This requirement is in contrast to most drawdown constraints in the literature.  Specifically, most drawdown constraints apply to the surplus or wealth, not to rates of payout or consumption; see, for example, \cite{GrossmanZhou1993}, \cite{CvitanicKaratzas1995}, and \cite{ElieTouzi2008}.

That said, there are two important papers that do impose drawdown constraints on consumption. The first is \cite{Dybvig1995}, who imposes a ratcheting constraint on the rate of consumption and finds the optimal investment and consumption policies for an investor in a Black-Scholes financial market who seeks to maximize discounted utility of consumption, in which risk preferences exhibit constant relative risk aversion, as we assume in this paper.  The second is \cite{Arun2012}, who extends \cite{Dybvig1995} by allowing the rate of consumption to decrease, but not below a fraction of its maximum rate.  This constraint is identical to the drawdown constraint that we apply to the excess dividend rate.  What distinguishes our work from \cite{Dybvig1995} and \cite{Arun2012} is that we do \textit{not} impose the additional requirement that wealth remain non-negative; instead, we allow bankruptcy, which occurs with positive probability, and which more closely follows the dividend models used in the literature (see, for example, \cite{GerberShiu2004}).  We further discuss the difference between our work and theirs in Remark \ref{rem:DybvigArun}.

\cite{ABB2018} study a problem related to the one considered by \cite{Dybvig1995}, except they allow a one-time only increase in the dividend rate.  They pre-specify the two dividend rates and determine the optimal level of surplus above which the company pays out at the higher rate.  On the other hand, we allow the dividend rate to increase without restriction, other than it may not drop below a fixed proportion of its historical maximum.  Therefore, our work is much more general than that of \cite{ABB2018}.

The remainder of the paper is organized as follows.  In Section \ref{sec:2}, we define our optimal dividend problem.
In Section \ref{sec:4}, we hypothesize that the value function is
a solution of a free-boundary problem with free-boundary conditions arising from smooth-fit and super-contact conditions and with two state variables.  In Section \ref{sec:41}, we reduce the dimension of our free-boundary problem from two state variables to one and use the convex Legendre transform to solve the dual of the resulting free-boundary problem; then, in Section \ref{sec:42}, we reverse the Legendre transform to obtain our value function. 
In Section \ref{sec:Prop}, we prove further properties of the optimal investment and dividend policies, and we consider limiting cases of our problem. In Sections \ref{sec:4} and \ref{sec:Prop}, we also present some numerical examples to demonstrate our results. Appendix \ref{app:Verify} provides the verification argument for our stochastic control problem, namely, that the solution of the free-boundary problem obtained in Section \ref{sec:4} is the value function.

%-----------------------------------------------------------------------------------
%
%       SECTION: 		Optimal dividend problem
%
%-----------------------------------------------------------------------------------

\section{Preliminary definitions and the problem setup}\label{sec:2}

Consider a company that has to decide on its investment and dividend policies. For simplicity, we assume that the number of shares of the company is fixed. Thus, its investment policy is dictated by its debt policy, that is, how much bond it issues or buys back. The bonds are issued at a fixed interest rate $r \ge 0$.  We represent the investment policy via the value of the total assets of the company at time $t$, denoted by $\pi_t$.  We assume that the company can instantly increase its total assets by issuing bonds and buying new assets.  Similarly, it can instantly reduce its total assets by selling existing assets and using the proceeds to buy back its bonds.

The company also chooses how to pay dividends to its shareholders. Let $C_t$ denote the rate at which the company pays dividends at time $t$; therefore, the total amount of dividends paid over $[t, t + \eps]$ is $\int_t^{t+\eps}C_u \dd u$.

We assume that the company is subject to two constraints when devising its dividend policy.  First, shareholders expect a risk premium for investing in the firm. Therefore, the dividend rate must be at least as high as the interest rate, that is, $C_t \ge rX_t$ for all $t \ge 0$, in which $X_t$ is the company's surplus at time $t$.  We denote the \textit{excess dividend rate} by
\begin{align}
	c_t = C_t - rX_t; \quad t \ge 0.
\end{align}

Second, because shareholders and analysts react negatively (and arguably overreact) when the rate of dividend payment decreases, we assume that the excess dividend rate cannot go below a fraction $\al \in (0, 1)$ of its past maximum, a so-called \textit{drawdown} constraint on the excess dividend rate.  Specifically, we impose the drawdown constraint
\begin{align}\label{eq:DDConst}
	c_t \ge \al z_t, \;\; \Pb \text{-almost surely}; \quad t \ge0,
\end{align}
in which $(z_t)_{t \ge0}$ is the historical peak of the dividend process, given by
\begin{equation}\label{eq:z}
	z_t = \max \big\{z, \sup_{0 \le s < t} c_s \big\}.
\end{equation}
Here, the constant $z > 0$ represents the historical maximum of the excess dividend rate strictly before time $0$ and is included so that the problem has a financial past. In particular, the drawdown constraint yields
\begin{align}
	c_t \ge \al z_t \ge \al z > 0 \quad \Longrightarrow \quad C_t > rX_t,
\end{align}
$\Pb$-almost surely, for all $t \ge 0$.  In other words, we assume that shareholders do not accept dividend rates at or below the amount they could earn at the risk-free rate. 

\begin{remark}
At any time $t \ge 0$, we allow the excess dividend rate to increase beyond its historical peak. If that were the case, we would have $c_t > z_t$. \qed
\end{remark}

Let $(I_t)_{t \ge 0}$ denote the intrinsic value of the company's total assets. Specifically, $I_t$ is the total assets of the company at $t$ assuming that it it has the total assets of \$1 at $t = 0$, all in equity with no debt, and assuming it does not pay dividends during $[0, t]$. We assume that the intrinsic value follows a geometric Brownian motion; specifically,
\begin{equation}
	\frac{\dd I_t}{I_t} = (\mu + r)dt + \sig \dd W_t,
\end{equation}
for some constants $\mu > 0$ and $\sig > 0$.  Here, $(W_t)_{t \ge 0}$ is a standard Brownian motion on a filtered probability space, $(\Omega, \mathcal{F}, (\mathcal{F}_t)_{t \ge 0}, \Pb)$, in which $(\mathcal{F}_t)_{t \ge 0}$ is the filtration generated by the Brownian motion and satisfies the usual conditions.

Given an investment policy $(\pi_t)_{t \ge 0}$ and an excess dividend policy $(c_t = C_t - rX_t)_{t\ge0}$, the net surplus process $(X_t)_{t \ge 0}$ is given by
\begin{equation}\label{eq:X}
\dd X_t = \pi_t \, \frac{\dd I_t}{I_t} + \Big(r(X_t-\pi_t)- C_t\Big) \dd t = \Big(\mu \pi_t - c_t\Big) \dd t + \sig \pi_t \, \dd W_t,
\end{equation}
with $X_0 = x \ge 0$. A pair of investment and dividend policies $(\pi_t, c_t)_{t \ge 0}$ is \emph{admissible} if it satisfies the following conditions,
\begin{enumerate}
	\item[(i)] $(\pi_t)_{t\ge0}$ is $(\mathcal{F}_t)$-progressively measurable, $\pi_t \ge 0$ for all $t \ge 0$, and $\int_0^\infty \pi_t^2 \dd t < \infty$, $\Pb$-almost surely,
	
	\item[(ii)] $(c_t)_{t \ge 0}$ is $(\mathcal{F}_t)$-adapted, non-negative, and right-continuous with left limits; and,
	
	\item[(iii)] $(c_t)_{t \ge 0}$ satisfies the drawdown constraint \eqref{eq:DDConst}.

\end{enumerate}
Let $\Cb(\al, z)$ denote the set of all admissible investment and dividend policies.

\smallskip

For future reference, we also introduce one additional set of policies. The set of \emph{unconstrained policies} $\Cb_0$ is the set of all investment and dividend policies $(\pi_t,c_t)$ that satisfy Conditions (i) and (ii) above, that is, we do not enforce the drawdown constraint (iii). Note that $\Cb_0 = \lim_{\al \to 0+} \Cb(\al, z)$.

We assume the company wishes to maximize the expectation of the discounted utility of the dividends it pays, in excess of the risk-free interest, between now and when the net surplus reaches bankruptcy. In particular, let $\tau$ denote the time of bankruptcy, that is,
\begin{equation}\label{eq:tau}
	\tau = \inf \big\{ t \ge 0 : X_t \le 0 \big\}.
\end{equation}
The time of ruin $\tau$ in \eqref{eq:tau} depends on the pair of admissible investment and dividend policies used to control $X$ in \eqref{eq:X}, but for simplicity of notation, we write $\tau$ instead of $\tau^{X^{(\pi_t, c_t)}}$.  We consider the following objective for the company,
\begin{equation}\label{eq:Objective}
\sup_{(\pi_t,c_t) \in \Cb(\al, z)} \Eb \[ \int_0^\tau \ee^{-\del t} \, \frac{c_t^{1-p}}{1-p} \, \dd t \].
\end{equation}
The constant $\del > 0$ denotes the subjective time preference parameter that represents the desirability of paying dividends sooner rather than later. The constant $p$ represents the shareholders' constant relative risk aversion, and we assume that
\begin{equation}\label{eq:p_constraint}
	 \frac{1}{1 + \kap \del} < p < 1,
\end{equation}
in which $\kap$ is defined by
\begin{equation}\label{eq:kap}
\kap = \frac{2 \sig^2}{\mu^2}  .
\end{equation}
The feasibility constraint on $p$ in \eqref{eq:p_constraint}, namely, $p > \frac{1}{1 + \kap \del}$, is common in the literature and dates back to \cite{Merton1969}.   Furthermore, if $p > 1$, then $\frac{c^{1 - p}}{1 - p}$ is negative and increasing in $c$; thus, \eqref{eq:Objective} favors large values of $c$ and small values of $\tau$.  It follows that immediate liquidation of the firm is the optimal policy when $p > 1$.

Note also, that in the objective \eqref{eq:Objective}, we essentially assume that, when the net surplus reaches $0$ at time $\tau$, then the (excess) consumption rate $c_t = C_t = 0$ with probability one for all $t > \tau$. This form of the objective function is consistent with $\int_0^\infty \ee^{-\del t} \, \frac{c_t^{1-p}}{1-p} \, \dd t = \int_0^\tau \ee^{-\del t} \, \frac{c_t^{1-p}}{1-p} \, \dd t$ for $p$ in the range given by \eqref{eq:p_constraint}. However, it implies that the drawdown constraint is violated after bankruptcy.

\begin{remark}\label{rem:DybvigArun}
It is not possible to compare our optimal strategy with that of \cite{Dybvig1995} or \cite{Arun2012} because there is no common ground between our model and theirs.  Since we are working with excess dividend rates, such common ground would be the case of zero interest rate, that is, $r = 0$. However, \cite{Dybvig1995} and \cite{Arun2012} explicitly exclude the case of $r = 0$. The reason is that both studies \emph{mandate} that the wealth process $X$ must always be positive (that is, $\tau = \infty$ with probability $1$) by making sure that there is always the possibility of funding the dividend payment through the risk-free investment. To enforce this assumption, they imposed an upper bound on feasible consumption processes of the form $\sup_t C_t = Z_t\le \frac{r}{\alpha} X_t$, such that consumption can always be funded by risk-free investment. This constraint is clearly unreasonable for $r=0$ because it would mean that the only feasible consumption is $C\equiv0$.
	
By contrast with \cite{Dybvig1995} and \cite{Arun2012}, we do not impose such upper bound on the consumption process. Instead, we stop the problem when $X$ hits 0. Indeed, since we work with excess dividend processes, our result does not depend on $r$ at all and would remain valid and unchanged for $r = 0$.  \qed
\end{remark}

\section{Optimal dividend policy under the drawdown constraint}\label{sec:4}

The value function corresponding to the stochastic control problem \eqref{eq:Objective} is
\begin{equation}\label{eq:V}
V(x, z) = \sup_{(\pi_t, c_t) \in \Cb(\al, z)} \E^x \[ \int_0^\tau \ee^{-\del t} \, \frac{c_t^{1-p}}{1-p} \, \dd t \],
\end{equation}
for $(x, z) \in \Rb_+^2$, in which $\Eb^x$ denotes expectation conditional on $X_0 = x$.

Our main goal in this section is to identify $V(x,z)$ and the corresponding optimal investment and dividend policies.
% In this section, we identify the optimal investment and (excess) dividend policies under the drawdown constraint.
The solution relies upon the relationship between the surplus at time $t$, namely, $X_t$, and the historical peak of the dividend process, namely, $z_t = \max \big\{z, \, \sup_{0 \le s < t} c_s \big\}$.
 
Each time the surplus level reaches a new maximum, the company faces a decision whether or not to increase the excess dividend rate beyond its historical peak.  Increasing the excess dividend rate might add value by increasing the expected value of utility of the excess dividend stream. But, doing so raises the bar for the compulsory minimum excess dividend rate and decreases the time to bankruptcy.

We hypothesize that there exists a critical surplus-to-historical peak ratio $\ws$ such that if $X_t > \ws z_t$, then the company will immediately raise its excess dividend rate to $c_t = X_t/\ws$.  Otherwise, if $X_t \le \ws z_t$, then the company will only increase its excess dividend rate to maintain $X_s \le \ws z_s$ for all $s \ge t$.  Additionally, if $0 < X_t < \ws z_t$, then the company will pay dividends at a rate lying in the interval $[\al z_t, z_t]$.

% From Theorem \ref{lem:verif} and Corollary \ref{cor:Optimal} (and as we will show in the proof of Theorem \ref{thm:V})
As we will verify in the proof of Theorem \ref{thm:V} below, if we find a classical solution $v(x,z)$ of the following free-boundary problem (FBP) on $\Dc = \{ (x, z) \in \Rb_+^2: x \le z \ws, z > 0 \}$, that is increasing and concave in $x$, then $v$ equals the value function $V$:
\begin{equation}\label{eq:HJB-FB}
\begin{cases}
\del v = \max \limits_{\pi \in \Rb} \[ \mu \pi v_x + \dfrac{1}{2} \, \sig^2 \pi^2 v_{xx} \] + \max \limits_{\al z \le c \le z} \[ \dfrac{c^{1-p}}{1-p} - c v_x \], \\% \qquad 0 < x < \ws z, \\
v(0, z) = 0, \\
v_z(\ws z, z) = 0 = v_{xz}(\ws z, z).
\end{cases}
\end{equation}
The additional requirement that $v_{xz}(\ws z, z) = 0$ is a so-called {\it super-contact} condition. This condition ensures the optimality of the boundary $\ws$.  Indeed, from \cite{Dixit1991} and \cite{Dumas1991}, we have the smooth-pasting condition $v_{z}(\hat w z, z) = 0$ for the value function that corresponds to the barrier strategy defined by {\it any} value of $\hat w > 0$ in place of $\ws$, but for a barrier strategy to be {\it optimal}, we must impose the higher-order condition $v_{xz}(\ws z, z) = 0$.  From this higher-order condition, we deduce that the Hamilton-Jacobi-Bellman (HJB) equation in \eqref{eq:HJB-FB} equals 0 along the ray $(x, z) = (\ws z, z)$ for $z > 0$.  See Dixit (1991) and Dumas (1991) for further discussion of value-matching, smooth-pasting (such as $v_z (\ws z, z) = 0$), and super-contact conditions.

\begin{figure}[t!]
\centerline{
% \fbox{
\adjustbox{trim={0.05\width} {0.05\height} {0.05\width} {0.2\height},clip}{\includegraphics[scale=.4,page=1]{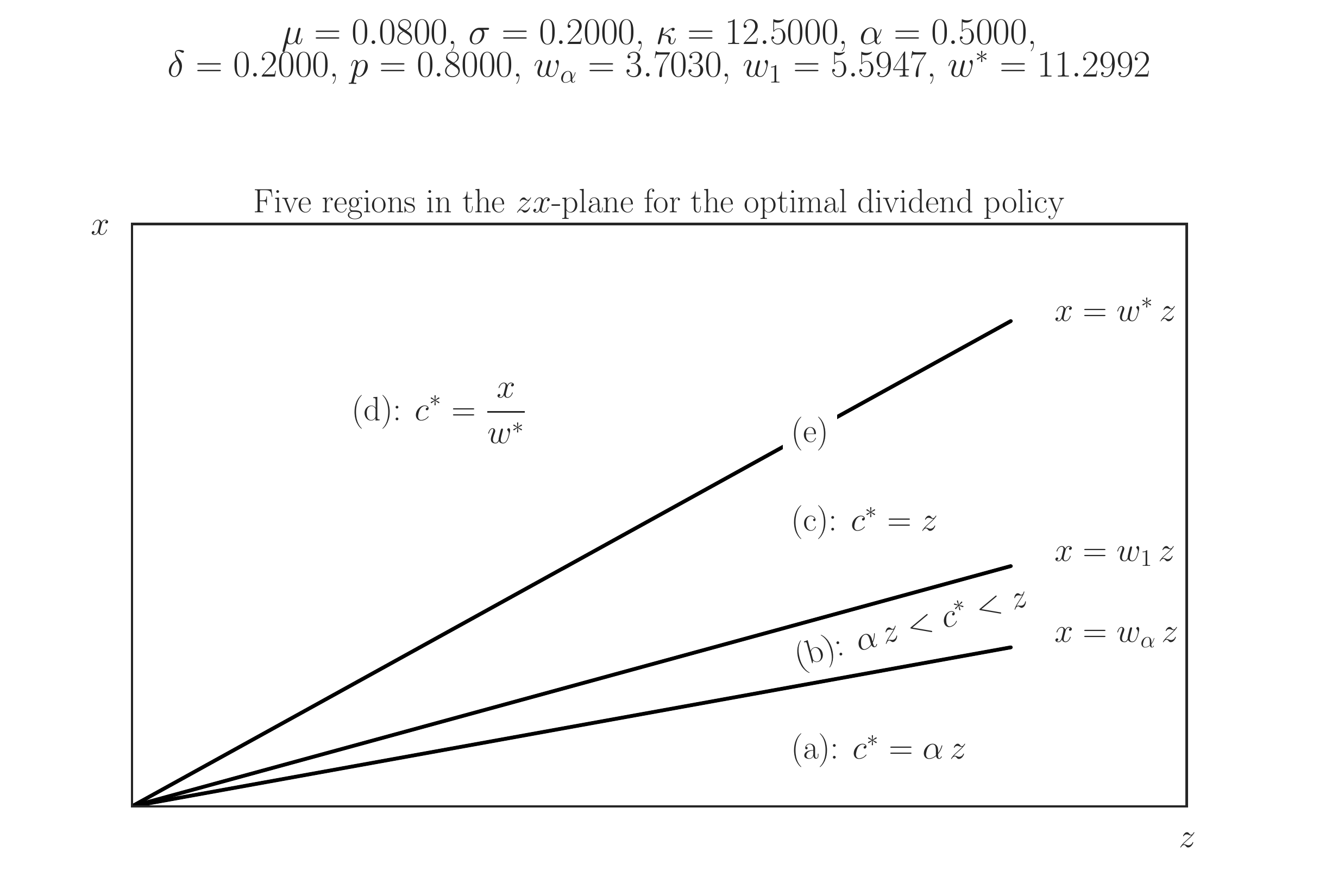}}
% }
}
\caption{A representative plot for the regions (a)-(e) in the $z-x$ plane showing different behavior of the optimal dividend policy. The optimally controlled process $(z_t^*, X_t^*)$ can only be in region (d) at $t = 0$, in which case it will be immediately moved to line (e). Thereafter, the optimally controlled process stays in regions (a), (b), (c), and (e).  Region (e) acts as a reflective barrier, in the sense that the optimal controlled processes are stopped from passing through by singular control at (e).}
\label{fig:XYregion}
\end{figure}

Furthermore, we hypothesize that there exist two other important surplus-to-historical peak ratios, namely $\wa$ and $\wo$ such that $0 < \wa < \wo < \ws$.  Based on these critical values, we hypothesize that the optimal dividend policy has the following structure:\footnote{For simplicity of notation, we omit the superscript $*$ that indicates the processes are the optimally controlled ones.  So, throughout this ansatz, $X_t = X^*_t$, $z_t = z^*_t$, and $c_t = c^*_t$.  At times, we also omit the word ``excess'' when we refer to the dividend rate; however, for the remainder of the paper, ``dividend rate'' means $c_t$, the ``\textit{excess} dividend rate.''}
\begin{enumerate}
	\item[(a)] If $X_t < \wa z_t$, then $c_t = \al z_t$. In other words, if the surplus level is ``very low,'' it is optimal for the firm to pay dividends at the lowest rate permitted by the drawdown constraint \eqref{eq:DDConst}.
	
	\item[(b)] If $\wa z_t < X_t < \wo z_t$, then $c_t = c^*(X_t, z_t) \in (\al z_t, z_t)$, for some function $c^*(x, z)$. In this case, the surplus is at an ``intermediate level,'' and the company distributes dividends at a rate greater than that of the minimum possible rate, but lower than the historical peak.
	
	\item[(c)] If $\wo z \le X_t < \ws z$, then $c_t = z_t$. In this case, surplus is large enough so that it is optimal to pay dividends at the historical peak but not large enough to raise the excess dividend rate above this value.
	
	\item[(d)] If $ X_t > \ws z_t$, then $c_t = \frac{X_t}{\ws} > z_t$.  In other words, if the surplus level is ``very large,'' the company will pay dividends at a rate greater than the historical peak.  Note that, in this case, the historical peak has a jump at $t$, that is, $\lim_{s \to t^+} z_s = \frac{X_t}{\ws} > z_t$.
	
	\item[(e)] Along the line $x = \ws z$, the company increases its dividend rate via singular control to keep $X_t \le \ws z_t$.
\end{enumerate}

By following the dividend policy above, the historical peak $(z_t)$ can have a jump only at time $t = 0$ and only if rule (d) is applicable, that is, $X_0 > \ws z_0$.  After this possible initial jump, the process $(X_t, z_t)_{t \ge 0}$ will be kept in the domain
\begin{equation}\label{eq:Dc}
	\Dc = \{(x, z): 0 \le x \le \ws z, z > 0 \},
\end{equation}
by applying the above rules.  In particular, after a possible initial jump, $(z_t)$ is only allowed to increase via singular control in order to keep $(X_t, z_t)$ inside $\Dc$.

% \mcomment{Bahman:\\Add an illustration of the optimal policy here.}

\medskip

\begin{remark}\label{rem:MStar}
	Define the running maximum of the surplus process by
	\begin{align}\label{eq:M}
	\begin{cases}
		M_0 = \ws z_0,\\
		\displaystyle
		M_t = \max \Big\{ M_0, ~ \max_{0 \le s <  t} X_s \Big\}; \quad t > 0.
	\end{cases}	
	\end{align}
By the discussion above, we have $M_t = \ws z_t$, for all $t \ge 0$. In other words, the running maximum of surplus, as defined in \eqref{eq:M}, is proportional to the historical peak of the excess dividend rate.  \qed
\end{remark}

We have two main tasks ahead of us.  First, we need to specify the unknowns $\wa$, $\wo$, $\ws$, and $c^*(x,z)$, along with the optimal investment policy $\pi^*(x, z)$. Second, we need to prove that the dividend policy hypothesized above is optimal.

% \begin{wrapfigure}{R}{0.5\textwidth}
% \hfill
% \fbox{\adjustbox{trim={0.05\width} {0\height} {0\width} {0\height},clip}{\includegraphics[scale=.5,page=1]{Phihatz.pdf}}}
% \end{wrapfigure}

\subsection{Reducing the dimension and applying the Legendre transform} \label{sec:41}

The value function $V$ in \eqref{eq:V} is homogeneous of degree $1 - p$ with respect to $x$ and $z$, that is,
\begin{equation}
V(\beta x, \beta z) = \beta^{1-p} V(x, z),
\end{equation}
for all $\beta > 0$.  Thus, if we define the function $U$ by
\begin{equation}\label{eq:U}
U(w) = V(w, 1),
\end{equation}
for $w \ge 0$, then we can recover $V$ from $U$ by
\begin{equation}\label{eq:V_U}
V(x, z) = z^{1-p} U(x/z),
\end{equation}
for all $(x, z) \in \Rb_+^2$.

From \eqref{eq:HJB-FB} and \eqref{eq:V_U}, we deduce the following FBP for $U$ for $w \in [0, \ws]$, in which we identify $V$ with $v$, the solution of \eqref{eq:HJB-FB}:
\begin{equation}\label{eq:HJB-U}
\begin{cases}
\del U = \max \limits_{\pih \in \Rb} \[ \mu \pih U_w + \dfrac{1}{2} \, \sig^2 \pih^2 U_{ww} \] + \max \limits_{\al \le \ch \le 1} \[ \dfrac{\ch^{1-p}}{1-p} - \ch U_w \], \\
U(0) = 0, \\
(1 - p) U(\ws) - \ws U_w(\ws) = 0, \\
p U_w(\ws) + \ws U_{ww}(\ws) = 0.
\end{cases}
\end{equation}
After we obtain $\pih^*$ and $\ch^*$, then we will be able to get $\pi^*$ and $c^*$ for $V$'s problem via $\pi^*(x, z) = \pih^*(x/z) z$ and $c^*(x, z) = \ch^*(x/z) z$.

Because $V$ is increasing and concave with respect to $x$, $U$ is increasing and concave with respect to $w$; thus, we  rewrite the HJB equation in \eqref{eq:HJB-U} as follows:
\begin{equation}\label{eq:HJB-U_de}
\dfrac{1}{\kap} \, \dfrac{U_w^2}{U_{ww}} + \del U =
\begin{cases}
\dfrac{\al^{1-p}}{1-p} - \al U_w, &\quad 0 \le w \le \wa, \vspace{1ex} \\
\dfrac{p}{1-p} \, \big(U_w(w) \big)^{-\frac{1-p}{p}}, &\quad \wa < w < \wo, \vspace{1ex} \\
\dfrac{1}{1-p} - U_w, &\quad \wo \le w \le \ws,
\end{cases}
\end{equation}
in which $\kap$ is defined in \eqref{eq:kap}.  Because of the non-linear terms $U_w^2/U_{ww}$ and $\big(U_w(w) \big)^{-\frac{1-p}{p}}$ in \eqref{eq:HJB-U_de}, it is natural to apply the Legendre transform to linearize this differential equation.  Specifically, define the dual variable $y$ and the corresponding convex dual function $\Uh$ by $y = U_w$ and $\Uh(y) = U(w) - wy$, respectively.  Also, define $y_0 = U_w(0)$ and $\ys = U_w(\ws)$.  Then, the differential equation in \eqref{eq:HJB-U_de} becomes the following \textit{linear} differential equation:
\begin{equation}\label{eq:HJB-U_de_y}
y^2 \Uh_{yy} + \kap \del y \Uh_y - \kap \del \Uh =
\begin{cases}
\kap  \Bigg(\al y - \dfrac{\al^{1-p}}{1-p} \Bigg), &\quad \al^{-p} \le y \le y_0, \vspace{1ex} \\
- \, \dfrac{\kap p}{1-p} \, y^{-\frac{1-p}{p}}, &\quad 1 < y < \al^{-p}, \vspace{1ex} \\
\kap \Bigg( y - \dfrac{1}{1-p} \Bigg), &\quad \ys \le y \le 1.
\end{cases}
\end{equation}
In exchange for linearity, the boundary condition $U(0) = 0$ becomes the free-boundary condition $\Uh(y_0) = 0 = \Uh_y(y_0)$ for the unknown boundary $y_0 > \al^{-p}$.  Furthermore, the smooth-pasting and super-contact conditions in \eqref{eq:HJB-U} become, respectively,
\begin{equation}\label{eq:Uh_ys1}
(1 - p) \Uh(\ys) + p \ys \Uh_y(\ys) = 0,
\end{equation}
and
\begin{equation}\label{eq:Uh_ys2}
\Uh_y(\ys) + p \ys \Uh_{yy}(\ys) = 0,
\end{equation}
at the unknown free-boundary $0 < \ys < 1$.

In the following proposition, we give the solution of this free-boundary problem.

\begin{proposition}\label{prop:Uh}
Suppose $\frac{1}{1 + \kap \del} < p < 1$.  The solution $\Uh$ of the differential equation \eqref{eq:HJB-U_de_y} subject to the free-boundary conditions $\Uh(y_0) = 0 = \Uh_y(y_0)$ and equations \eqref{eq:Uh_ys1} and \eqref{eq:Uh_ys2} is given by
\begin{equation}\label{eq:Uh}
\Uh(y) = 
\begin{cases}
C_1 y + C_2 y^{-\kap \del} + \dfrac{\kap \al}{1 + \kap \del} \, y \ln y + \dfrac{\al^{1-p}}{\del(1 - p)} \, , &\quad \al^{-p} \le y \le y_0, \vspace{1ex} \\
C_3 y + C_4 y^{-\kap \del} + 
\dfrac{\kap}{1 - p} \, \dfrac{p^3}{p(1 + \kap \del) - 1} \, y^{-\frac{1-p}{p}},
&\quad 1 < y < \al^{-p}, \vspace{1ex} \\
C_5 y + C_6 y^{-\kap \del} + \dfrac{\kap}{1 + \kap \del} \, y \ln y + \dfrac{1}{\del(1 - p)} \, , &\quad \ys \le y \le 1, 
\end{cases}
\end{equation}
in which
\begin{align}
C_1 &= - \, \dfrac{\kap \al}{1 + \kap \del} \( \ln \etas - p \ln \al + \dfrac{1}{\etas(1 - p)} + \dfrac{1}{1 + \kap \del} \), \label{eq:C1} \vspace{1ex} \\
C_2 &= \dfrac{\al^{1- p(1 + \kap \del)}}{1 + \kap \del} \( \dfrac{\kap}{1 + \kap \del} \, (\etas)^{1 + \kap \del} - \dfrac{1}{\del(1 - p)} \, (\etas)^{\kap \del} \) > 0, \label{eq:C2} \vspace{1ex} \\
C_3 &= - \, \dfrac{\kap \al}{1 + \kap \del} \( \ln \etas + \dfrac{1}{\etas(1 - p)} - (1 + p) \), \label{eq:C3} \vspace{1ex} \\
C_4 &= \dfrac{\al^{1- p(1 + \kap \del)}}{1 + \kap \del} \( \dfrac{\kap}{1 + \kap \del} \, (\etas)^{1 + \kap \del} - \dfrac{1}{\del(1 - p)} \, (\etas)^{\kap \del} - \dfrac{1}{\del (1 + \kap \del) \big( p(1 + \kap \del) - 1 \big)} \) < 0, \label{eq:C4} \vspace{1ex} \\
C_5 &=  - \, \dfrac{\kap}{1 + \kap \del} \( \al \ln \etas + \dfrac{\al}{\etas(1 - p)} + (1- \al)(1 + p) + \dfrac{1}{1 + \kap \del} \),  \label{eq:C5} \\
C_6 &= \dfrac{\al^{1- p(1 + \kap \del)}}{1 + \kap \del} \( \dfrac{\kap}{1 + \kap \del} \, (\etas)^{1 + \kap \del} - \dfrac{1}{\del(1 - p)} \, (\etas)^{\kap \del} \) - \dfrac{\al^{1- p(1 + \kap \del)} - 1}{\del (1 + \kap \del)^2 \big( p(1 + \kap \del) - 1 \big)} > 0 ,  \label{eq:C6}
\end{align}
and $\etas = y_0 \al^p > 1$ and $0 < \ys < 1$ uniquely solve the following system of two equations:
\begin{equation}\label{eq:etas_ys}
\begin{cases}
\ln \dfrac{\eta^\al}{y} + \dfrac{\al}{\eta(1 - p)} - \dfrac{1}{y} = \al(1 + p) - 1, \vspace{1ex} \\
\al^{1- p(1 + \kap \del)} \big( p(1 + \kap \del) - 1 \big) \(  \dfrac{\kap}{1 + \kap \del} \, \eta^{1 + \kap \del} - \dfrac{1}{\del(1 - p)} \, \eta^{\kap \del} \) + \( \dfrac{\kap}{1 + \kap \del} \, y^{1 + \kap \del} - \dfrac{1}{\del} \, y^{\kap \del} \) \\
\quad = \dfrac{\al^{1- p(1 + \kap \del)} - 1}{\del (1 + \kap \del)}.
\end{cases}
\end{equation}
Moreover, $\Uh$ is strictly decreasing and strictly convex with continuous second derivative on $(\ys, y_0)$.  \qed
\end{proposition}

\begin{proof}
The expression for $\Uh$ in \eqref{eq:Uh} and the values of $C_i$ for $i = 1, 2, \dots, 6$ in \eqref{eq:C1}-\eqref{eq:C6} follow readily by solving \eqref{eq:HJB-U_de_y}, by imposing the free-boundary condition at $y = y_0$, and by requiring that $\Uh$ have continuous derivative at the boundary points.  Furthermore, we get the two equations in \eqref{eq:etas_ys} by imposing the free-boundary condition at $y = \ys$, by solving for $C_5$ and $C_6$ in terms of $\ys$, and by then equating those expressions with the ones in \eqref{eq:C5} and \eqref{eq:C6}, respectively.  For the reader's convenience, we write $C_5$ and $C_6$ in terms of $\ys$.
\begin{equation}\label{eq:C5_ys}
C_5 =  - \, \dfrac{\kap}{1 + \kap \del} \( \ln \ys + \dfrac{1}{\ys} + p + \dfrac{1}{1 + \kap \del} \),
\end{equation}
and
\begin{equation}\label{eq:C6_ys}
C_6 = \dfrac{1}{\del(1 + \kap \del) \big( p(1 + \kap \del) - 1 \big)} \( (\ys)^{\kap \del} - \dfrac{\kap \del}{1 + \kap \del} (\ys)^{1 + \kap \del} \).
\end{equation}
Checking that $\Uh$ is strictly decreasing and strictly convex with continuous second derivative on $(\ys, y_0)$ is straightforward.

\begin{figure}[t!]
\centerline{
% \fbox{
\adjustbox{trim={0.01\width} {0.09\height} {0.6\width} {0.15\height},clip}{\includegraphics[scale=.5,page=1]{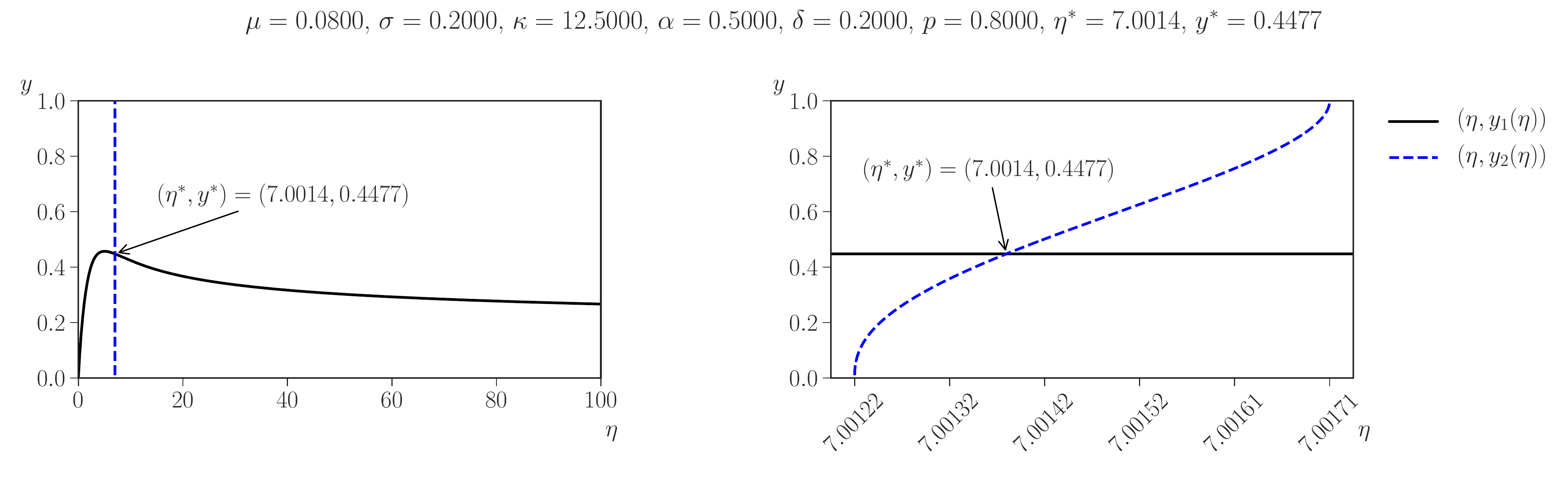}}
% }
% \fbox{
\adjustbox{trim={0.89\width} {0.09\height} {0.0\width} {0.15\height},clip}{\includegraphics[scale=.5,page=1]{EtaYSTAR.pdf}}
% }
% \fbox{
\adjustbox{trim={0.49\width} {0.09\height} {0.125\width} {0.15\height},clip}{\includegraphics[scale=.5,page=1]{EtaYSTAR.pdf}}
% }
}
\caption{A representative graph of the two curves defined by the equations in \eqref{eq:etas_ys}. The black solid curve (resp., the blue dotted curve) is implicitly defined by the first (resp., second) equation. The right plot is an enlargement of the left plot over the horizontal axis, which better shows the solution of the system of equations. The chosen values of the parameters are $\mu = 0.08$, $\sig = 0.2$, $\delta = 0.2$, $\al = 0.5$, and $p=0.8$.}
\label{fig:EtaY_star}
\end{figure}

It remains to show that the two equations in \eqref{eq:etas_ys} have a unique solution $\etas > 1$ and $0 < \ys < 1$.  To that end, consider the second equation in \eqref{eq:etas_ys}.  When $y = 0$, we can rewrite this equation as
\begin{equation}\label{eq:etay_y0}
 \(  \dfrac{\kap}{1 + \kap \del} \, \eta^{1 + \kap \del} - \dfrac{1}{\del(1 - p)} \, \eta^{\kap \del} \) - \dfrac{1 - \al^{p(1 + \kap \del) - 1}}{\del (1 + \kap \del)\big( p(1 + \kap \del) - 1 \big)} = 0.
\end{equation}
When $\eta = 1$, the left side of \eqref{eq:etay_y0} is negative.  Furthermore, as $\eta$ increases from $1$ to $\frac{1}{1-p}$, the left side decreases; then, as $\eta$ increases from $\frac{1}{1 - p}$ to $\infty$, the left side increases from a negative number to positive $\infty$.  Thus, when $y = 0$, there exists a unique solution $\eta >  1$ of \eqref{eq:etay_y0}; moreover, this unique solution is greater than $\frac{1}{1-p}$.

Next, differentiate the second equation in \eqref{eq:etas_ys} with respect to $y$, treating $\eta$ as a function of $y$, to obtain
\begin{equation}
\al^{1- p(1 + \kap \del)} \big( p(1 + \kap \del) - 1 \big) \eta^{\kap \del} \( 1 - \dfrac{1}{\eta(1 - p)} \) \dfrac{d \eta}{dy} = y^{\kap \del} \( \dfrac{1}{y} - 1 \).
\end{equation}
Thus, $\frac{d \eta}{dy} < 0$ for $1 < \eta < \frac{1}{1 - p}$, and $\frac{d \eta}{dy} > 0$ for $\eta > \frac{1}{1 - p}$.  We showed in the previous paragraph that, when $y = 0$, the unique solution of \eqref{eq:etay_y0} is greater than $\frac{1}{1 - p}$; thus, as $y$ increases, the unique solution $\eta$ of the second equation in \eqref{eq:etas_ys} increases.  It follows that we can restrict our attention to the range $\eta > \frac{1}{1 - p}$, on which the curve defined by the second equation in \eqref{eq:etas_ys} increases with respect to $y$ in the $y\eta$-plane, and as $y$ increases from $0$ to $1$, $\eta$ increases from some finite number greater than $\frac{1}{1 -p}$ to another finite number.

Finally, differentiate the first equation in \eqref{eq:etas_ys} with respect to $y$ to obtain
\begin{equation}
\dfrac{\al}{\eta} \( 1 - \dfrac{1}{\eta(1 - p)} \) \dfrac{d \eta}{dy} = \dfrac{1}{y} \( 1 - \dfrac{1}{y} \).
\end{equation}
Thus, for the first equation, $\frac{d \eta}{dy} < 0$ when $\eta > \frac{1}{1 - p}$.  Moreover, as $y$ increases from $0$ to some number $y_p < 1$, $\eta$ decreases from $\infty$ to $\frac{1}{1 - p}$.  In other words, this curve fully covers the range of $\eta$ from $\frac{1}{1 - p}$ to $\infty$.

It follows from the above observations that there exists a unique point $(\ys, \etas) \in (0, 1) \times (1, \infty)$ of intersection of the two curves in \eqref{eq:etas_ys}; moreover, $\ys < y_p$ and $\etas > \frac{1}{1 - p}$.
\end{proof}

In the following corollary, we show that $\Uh$ decreases with respect to $\al$, and we compute $\Uh$ as $\al \to 0+$; we will use this result in the proof of Theorem \ref{thm:V} to show that the candidate value function satisfies the so-called transversality condition, see Lemma \ref{lem:V_al}.

\begin{corollary}\label{cor:Uh_al}
$\Uh(y)$ is non-decreasing with respect to $\al$, for any fixed value of $y \in (\ys, y_0)$.
\end{corollary}

\begin{proof}
Clearly $\Uh$ is differentiable with respect to $\al$.  To find the differential equation that $\Uh_\al$ solves, differentiate \eqref{eq:HJB-U_de_y} with respect to $\al$ to obtain
\begin{equation}\label{eq:Uh_al}
y^2 (\Uh_\al)_{yy} + \kap \del y (\Uh_\al)_y - \kap \del \Uh_\al =
\begin{cases}
\kap \big(y - \al^{-p} \big), &\quad \al^{-p} \le y \le y_0, \vspace{1ex} \\
0, &\quad \ys \le y < \al^{-p} .
\end{cases}
\end{equation}
Thus, if we define $G$ by
\begin{equation}
G(y, \phi, \phi_y, \phi_{yy}) = \kap \del \phi - \kap \del y \phi_y - y^2 \phi_{yy},
\end{equation}
then $G$ is increasing in $\phi$ and decreasing in $\phi_{yy}$.  In other words, it satisfies the monotonicity requirement (0.1) in \cite{CrandallIshiiLions1992}.  Also, note that $G(y, \Uh_\al, (\Uh_\al)_y, (\Uh_\al)_{yy}) \le 0$; thus, if we show $\Uh_\al(\ys) \le 0$ and $\Uh_\al(y_0) \le 0$, then, from Theorem 3.3 of \cite{CrandallIshiiLions1992}, we can deduce that $\Uh_\al \le 0$ for $\al \in (0, 1)$ and all $\ys \le y \le y_0$.

To show $\Uh_\al(\ys) \le 0$, use the free-boundary conditions in \eqref{eq:Uh_ys1} and \eqref{eq:Uh_ys2}, along with the differential equation for $\Uh$ in \eqref{eq:HJB-U_de_y}, to compute
\begin{equation}
\Uh(\ys) = \dfrac{\kap p^2}{p(1 + \kap \del) - 1} \( \dfrac{1}{1 - p} - \ys \).
\end{equation}
Differentiate this expression with respect to $\al$ and cancel factors of $\frac{\partial \ys}{\partial \al}$ to obtain
\begin{equation}
\Uh_\al(\ys) = - \, \dfrac{\kap p^2}{p(1 + \kap \del) - 1} < 0.
\end{equation}
Similarly, the free-boundary condition at $y_0$, namely, $\Uh(y_0) = 0$ implies that $\Uh_\al(y_0) = 0$.  Thus, it follows that $\Uh(y)$ is non-increasing with respect to $\al$ for any fixed value of $\ys < y < y_0$.
\end{proof}

\begin{figure}[t!]
\centerline{
% \fbox{
\adjustbox{trim={0.06\width} {0.15\height} {0.675\width} {0.12\height},clip}{\includegraphics[scale=.4,page=1]{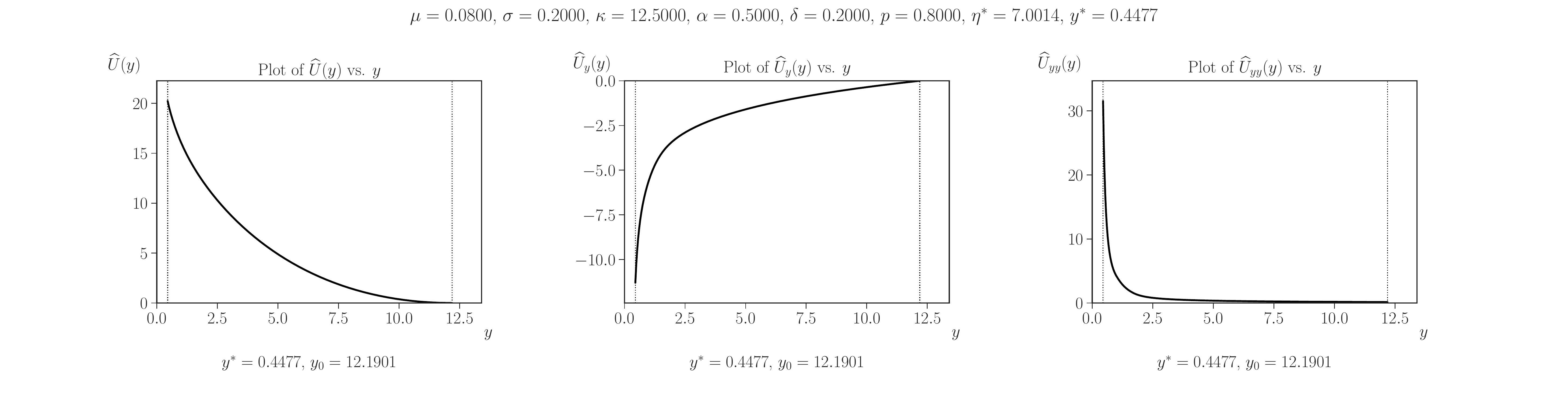}}
% }
%\fbox{
\adjustbox{trim={0.365\width} {0.15\height} {0.375\width} {0.12\height},clip}{\includegraphics[scale=.4,page=1]{Uhat.pdf}}
%}
%\fbox{
\adjustbox{trim={0.66\width} {0.15\height} {0.08\width} {0.12\height},clip}{\includegraphics[scale=.4,page=1]{Uhat.pdf}}
%}
}
\caption{A representative graph of $\Uh(y)$ and its first two derivatives. The parameters are as in Figure \ref{fig:EtaY_star}. The vertical dotted lines are $y=y^*=0.04477$ and $y=y_0=12.1901$, respectively.}
\label{fig:Uhat}
\end{figure}

Figure \ref{fig:Uhat} illustrates $\Uh$ given by \eqref{eq:Uh} for specific values of the parameters. Note that $\Uh$ is a decreasing and convex function, as expected. In the next section, we describe the solution of the HJB equation \eqref{eq:HJB-U} in terms of $\Uh$.

\subsection{Reversing the Legendre transform to obtain $V$ and the optimal policy}\label{sec:42}

Because we obtained $\Uh$'s FBP by applying the convex Legendre transform to $U$'s FBP, we now reverse this process and obtain $U$ as the concave Legendre transform of $\Uh$.  Specifically, $U(w) = \Uh(y) - y \Uh_y(y)$, in which $y \in [\ys, y_0]$ uniquely solves $w = - \Uh_y(y)$.  We are able to solve $w = - \Uh_y(y)$ because $\Uh$ is strictly decreasing and convex, and it follows that $U$ is increasing and concave, as expected.   The following theorem summarizes these (relatively straightforward) computations; in its statement, we rely on the notation of Proposition \ref{prop:Uh}.

\begin{theorem}\label{thm:V}
Suppose $\frac{1}{1 + \kap \del} < p < 1$. Let $\etas$ and $\ys$ be the solution of \eqref{eq:etas_ys} and define $y_0:=\eta^*\alpha^{-p}$.  Furthermore, define $\wa$, $\wo$, and $\ws$ by
\begin{equation}\label{eq:wa}
\wa = \dfrac{\kap \al}{1 + \kap \del} \left\{ \ln \etas + \( \dfrac{\kap \del}{1 + \kap \del} - \dfrac{1}{\etas(1 - p)} \) \( (\etas)^{1+ \kap \del}  - 1 \) \right\},
\end{equation}
\vskip 1pt
\begin{equation}\label{eq:wo}
\wo =  \dfrac{\kap}{1 + \kap \del} \left\{ \ln \ys + p + \( \dfrac{1}{\ys} - \dfrac{\kap \del}{1 + \kap \del} \) \( 1 + \dfrac{(\ys)^{1 + \kap \del}}{p(1 + \kap \del) - 1} \) \right\},
\end{equation}
and
\begin{equation}\label{eq:ws}
\ws = \frac{\kap p}{p(1 + \kap \del) -1} \left\{ \dfrac{1}{\ys} - (1 - p) \right\}.
\end{equation}
If $(x, z) \in \Dc$, that is, $0 \le x \le \ws z$, the value function $V$ in \eqref{eq:V} is given by \smallskip
\begin{equation}\label{eq:V-sol}
V(x, z) =
\begin{cases}
 \dfrac{\kap \al z^{1-p} y}{1 + \kap \del} \left\{ \( \dfrac{y_0}{y} \)^{1 + \kap \del} - 1 \right\}\\
 \hspace{2em}- \dfrac{(\al z)^{1 - p}}{\del(1 - p)} \left\{ \( \dfrac{y_0}{y} \)^{\kap \del} - 1 \right\}, & 0 \le x \le \wa z, \vspace{1.5ex} \\ 
\left\{ \dfrac{\kap \del p^2}{1 - p} \, y^{- \, \frac{1 - p}{p}} - \( \dfrac{\kap \del (\ys)^{1 + \kap \del}}{1 + \kap \del} - (\ys)^{\kap \del} + \dfrac{1}{1 + \kap \del} \) y^{-\kap \del} \right\}\\
 \hspace{2em}\times\dfrac{z^{1-p}}{\del \big(p(1 + \kap \del) - 1 \big)}, & \wa z < x < \wo z, \vspace{1.5ex} \\
\dfrac{z^{1-p}}{p(1 + \kap \del) - 1} \( \dfrac{1}{\del} - \dfrac{\kap \ys}{1 + \kap \del} \) \( \dfrac{\ys}{y} \)^{\kap \del} \\
 \hspace{2em}+ \dfrac{z^{1 - p}}{\del(1 - p)} - \dfrac{\kap z^{1 - p} y}{1 + \kap \del} \, , & \wo z \le x \le \ws z.
\end{cases}
\end{equation}
Here, $y \in [\ys, y_0]$ uniquely solves $- \Uh_y(y) = x/z$.

Moreover, when the optimally controlled surplus and peak excess dividend rate lie in $\Dc$, that is, $0 \le X^*_t \le \ws z^*_t$, then the optimal policies are given in feedback form by $( \pis(X^*_t, z^*_t), \cs(X^*_t, z^*_t) \big)$, in which
\begin{equation}\label{eq:pis}
\pis(x, z) = - \, \dfrac{\mu}{\sig^2} \, \dfrac{z U_w(x/z)}{U_{ww}(x/z)} = \dfrac{\mu}{\sig^2} \, zy \Uh_{yy}(y),
\end{equation}
and
\begin{equation}\label{eq:cs}
\cs(x, z) =
\begin{cases}
\al z, &\quad 0 \le x \le \wa z, \vspace{1ex} \\ 
y^{- \, \frac{1}{p}} \, z, &\quad \wa z < x < \wo z, \vspace{1ex} \\
z, &\quad \wo z \le x < \ws z.
\end{cases}
\end{equation}
Along the line $x = \ws z$, the company increases the excess dividend rate via singular control to keep $(X^*_t, z^*_t)$ within $\Dc$, that is, $X^*_t \le \ws z^*_t$.

If $x > \ws z$, then the company sets its initial investment policy to 
\begin{align}\label{eq:pis_xbig}
	\pis(x,z) = - \, \dfrac{\mu}{\sig^2} \, \dfrac{U_w(\ws)}{\ws U_{ww}(\ws)} \, x = \dfrac{\mu}{\sig^2} \, \dfrac{\ys}{\ws} \, \Uh_{yy}(\ys) \, x,
\end{align}
and immediately increases its excess dividend rate to 
\begin{equation}\label{eq:cs_xbig}
\cs(x, z) = \dfrac{x}{\ws} \, ,
\end{equation}
with corresponding value function
\begin{equation}\label{eq:V_xbig}
V(x, z) = \dfrac{x^{1 - p}}{1 - p} \, \dfrac{\kap p (\ws)^p}{\big(p(1 + \kap \del) - 1\big)\ws + \kap p(1 - p)} \, .
\end{equation}
Thereafter, the company increases the excess dividend rate only as needed to keep $X^*_t \le \ws z^*_t$.
\end{theorem}

\begin{proof}
	See Appendix \ref{app:Verify}.
\end{proof}

% \medskip

\begin{figure}[t!]
\centerline{
% \fbox{
\adjustbox{trim={0.06\width} {0.16\height} {0.675\width} {0.13\height},clip}{\includegraphics[scale=.4,page=1]{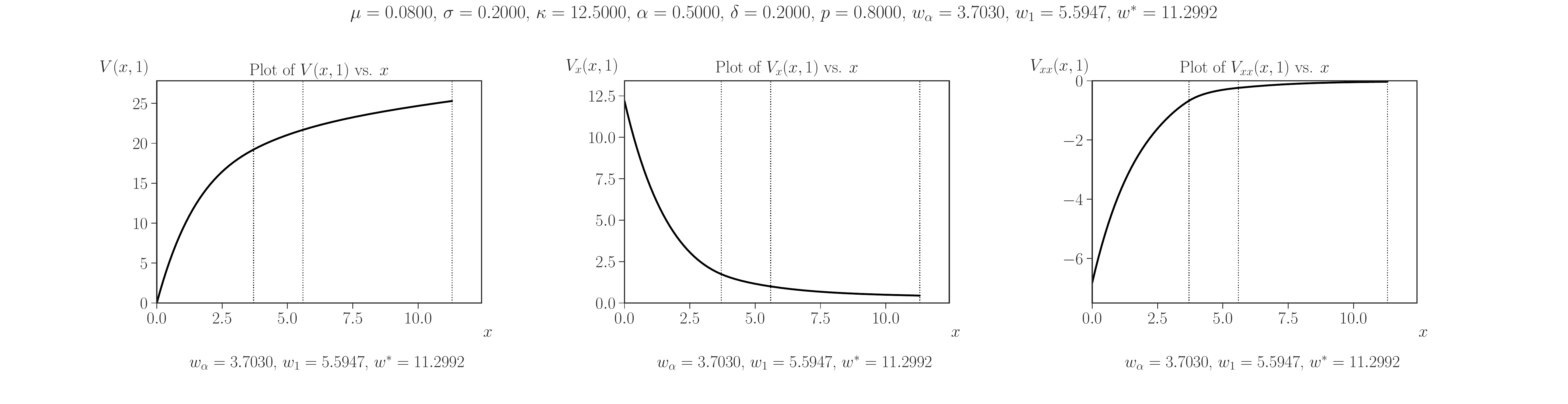}}
% }
%\fbox{
\adjustbox{trim={0.36\width} {0.16\height} {0.375\width} {0.13\height},clip}{\includegraphics[scale=.4,page=1]{U.pdf}}
%}
%\fbox{
\adjustbox{trim={0.655\width} {0.16\height} {0.08\width} {0.13\height},clip}{\includegraphics[scale=.4,page=1]{U.pdf}}
%}
}
\caption{A representative graph of $V(x, 1)$ and its first two derivatives with respect to $x$. The vertical dotted lines correspond to $x=\wa=3.703$, $x=\wo=5.5947$, and $x=\ws=11.2992$, respectively. The parameters are as in Figure \ref{fig:EtaY_star}.}
\label{fig:U}
\end{figure}

Figure \ref{fig:U} illustrates $V(x,1)=U(x)$ given by Theorem \ref{thm:V}, along with its derivatives with respect to $x$. Note that $U$ is the concave Legendre transform of $\Uh$ illustrated by Figure \ref{fig:Uhat}.

\begin{remark}
For the reader's reference, we now give the explicit equation solved by $y$, namely, $- \Uh_y(y) = x/z$:
\begin{align}\label{eq:y-sol}
\dfrac{x}{z} =
\begin{cases}
\dfrac{\kap \al}{1 + \kap \del} \left\{ \ln \dfrac{y_0}{y} + \( \dfrac{\kap \del}{1 + \kap \del} - \dfrac{\al^{-p}}{y_0(1 - p)} \) \( \( \dfrac{y_0}{y} \)^{1 + \kap \del} - 1 \) \right\}, & 0 \le x \le \wa z, \vspace{1.5ex} \\ 
\dfrac{\kap}{1 + \kap \del} \( \ln \ys + \dfrac{1}{\ys} - 1 \) + \dfrac{\kap p^2}{p(1 + \kap \del) - 1} \, y^{- \, \frac{1}{p}} \vspace{1.5ex} \\
\quad  - \, \dfrac{\kap y^{-(1 + \kap \del)}}{(1 + \kap \del) \big( p(1 + \kap \del) - 1 \big)} \left\{ \dfrac{\kap \del (\ys)^{1 + \kap \del}}{1 + \kap \del} - (\ys)^{\kap \del} + \dfrac{1}{1 + \kap \del} \right\} , & \wa z < x < \wo z, \vspace{1.5ex} \\
\dfrac{\kap}{1 + \kap \del} \left\{ \ln \dfrac{\ys}{y} + p + \( \dfrac{1}{\ys} - \dfrac{\kap \del}{1 + \kap \del} \) \( 1 + \dfrac{ \( \ys/y \)^{1 + \kap \del} }{p(1 + \kap \del) - 1}\) \right\}, & \wo z \le x \le \ws z .
\end{cases}
\end{align}
The intervals $0 \le x \le \wa z$, $\wa z < x < \wo z$, and $\wo z \le x \le \ws z$ correspond to $\al^{-p} \le y \le y_0$, $1 < y < \al^{-p}$, and $\ys \le y \le 1$, respectively.  Also, $\pis(x, z)$ is given by
\begin{equation}
\pis(x, z) =
\begin{cases}
\dfrac{2 \al z}{\mu} \left\{ \( \dfrac{\kap \del}{1 + \kap \del} - \dfrac{\al^{-p}}{y_0(1 - p)} \) \( \dfrac{y_0}{y} \)^{1 + \kap \del} + \dfrac{1}{1 + \kap \del} \right\}, & 0 \le x \le \wa z, \vspace{1ex} \\ 
\dfrac{2z}{\mu \big(p(1 + \kap \del) - 1\big)} \left\{ p y^{- \, \frac{1}{p}} - \dfrac{y^{-(1 + \kap \del)}}{1 + \kap \del} + \( \dfrac{1}{\ys} - \dfrac{\kap \del}{1 + \kap \del} \) \( \dfrac{\ys}{y} \)^{1 + \kap \del} \right\}, & \wa z < x < \wo z, \vspace{1ex} \\
\dfrac{2z}{\mu} \left\{ \dfrac{1}{1 + \kap \del} + \dfrac{1}{p(1 + \kap \del) - 1} \( \dfrac{1}{\ys} - \dfrac{\kap \del}{1 + \kap \del} \) \( \dfrac{\ys}{y} \)^{1 + \kap \del} \right\}, & \wo z \le x \le \ws z,
\end{cases}
\end{equation}
\end{remark}

\section{Properties of the optimal policy and the value function}
\label{sec:Prop}
In the final section of the paper, we investigate further properties of the optimal policy and the value function. 

In the next corollary, we examine how $\pis$ and $\cs$ change with respect to the state variables $x$ and $z$. Figures \ref{fig:piSTAR} and \ref{fig:cSTAR} illustrate these properties.

\begin{corollary}\label{cor:prop}
The optimal feedback functions $\pis = \pis(x, z)$ and $\cs = \cs(x, z)$ from \eqref{eq:pis} and \eqref{eq:cs}, respectively, satisfy the following properties.
\begin{enumerate}
%\item[$(a)$] If we allow $x \in (0, \ws z)$ to change so that $x/z$ is fixed, then $\pis$ and $\cs$ increase linearly with respect to $z$.

\item[$(a)$] For $\wa z < x < \wo z$ and $z > 0$, $\cs$ is increasing and convex with respect to $x$.

\item[$(b)$] For $0 < x < \wa z$ and $\wo z < x < \ws z$ and for $z > 0$, $\pis$ is increasing and convex with respect to $x$.

\item[$(c)$] For $\wa z < x < \wo z$ and $z > 0$, $\pis$ is increasing and concave with respect to $x$.

\end{enumerate}
\end{corollary}

\begin{proof}
For $\wa z < x < \wo z$, $\cs$ is proportional to $y^{- \, \frac{1}{p}}$.  Recall that $y = U_w$; thus,
\begin{equation}
\dfrac{\partial \cs}{\partial x} \propto - \, \dfrac{1}{p} \, y^{- \, \frac{1}{p} - 1} \, \dfrac{\partial y}{\partial x} \propto - y^{- \, \frac{1}{p} - 1} U_{ww} \propto \dfrac{y^{- \, \frac{1}{p} - 1}}{\Uh_{yy}} > 0.
\end{equation}
We have shown that $\cs$ increases with $x$.  Moreover,
\begin{align}
\dfrac{\partial^2 \cs}{\partial x^2} &\propto \dfrac{\partial}{\partial x} \dfrac{y^{- \, \frac{1}{p} - 1}}{\Uh_{yy}} \propto  \left\{- \, \dfrac{1 + p}{p} \, y^{- \, \frac{1}{p} - 2} \, \Uh_{yy}  - y^{- \, \frac{1}{p} - 1} \, \Uh_{yyy} \right\} \dfrac{\partial y}{\partial x} \\
&\propto \dfrac{1 + p}{p} \, \Uh_{yy}  + y \Uh_{yyy} \propto \big( (1+ p) - p(2 + \kap \del) \big) C_4 \propto p(1 + \kap \del) - 1 > 0,
\end{align}
because $C_4 < 0$.  We have shown that $\cs$ is convex in $x$.

% \mcomment{Bahman: Add plots of $\pis(\cdot,z)$ and $\cs(\cdot,z)$ for two values of $z$, say, $z=1$ and $z=10$.}

\begin{figure}[t!]
\centerline{
% \fbox{
\adjustbox{trim={0.04\width} {0.5\height} {0.0\width} {0.1\height},clip}{\includegraphics[scale=.4,page=1]{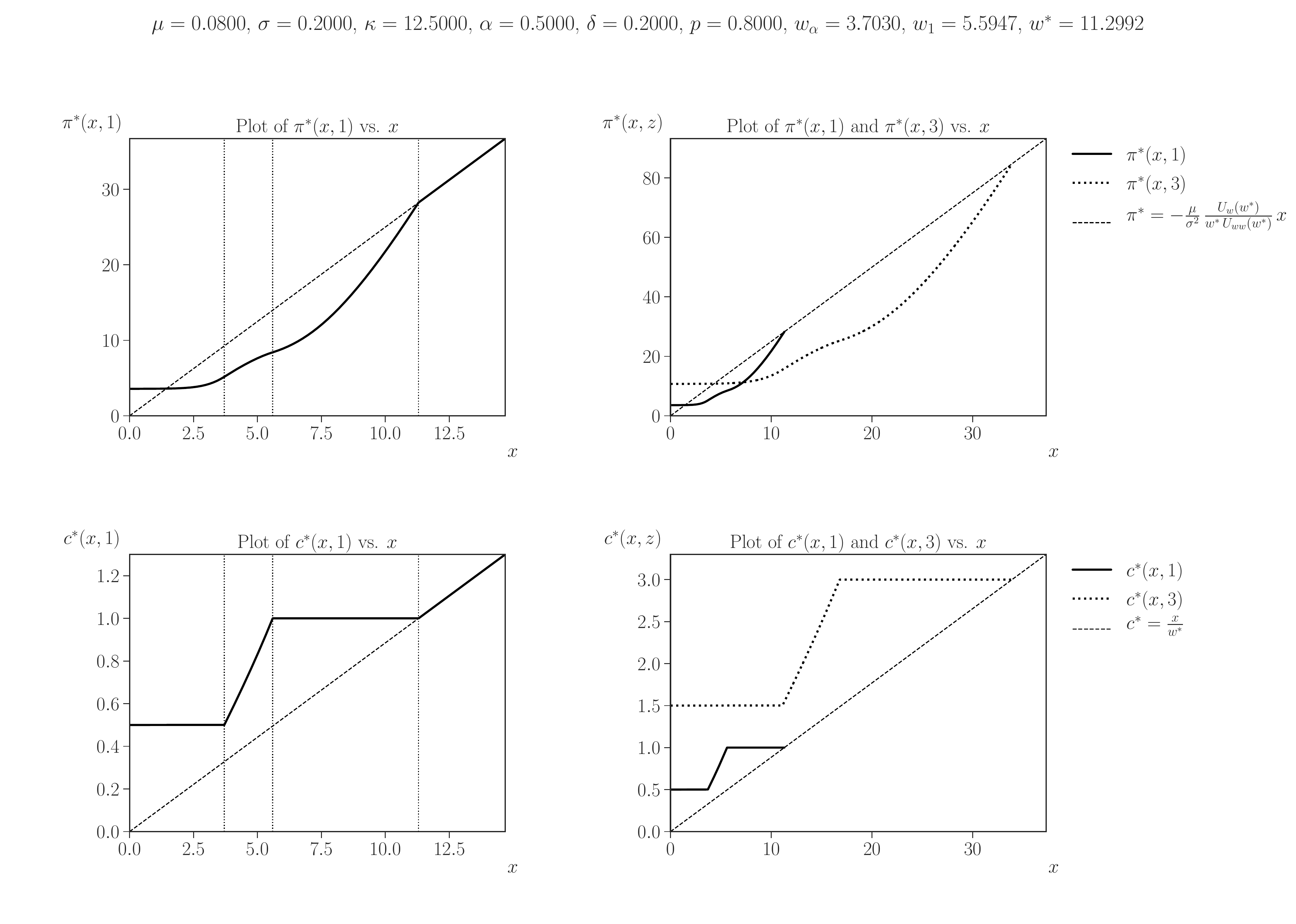}}
% }
}
\caption{Graph of $\pis(x, 1)$ on the left, and graphs of $\pis(x,1)$ and $\pis(x,2)$ on the right. In both graphs, the dashed diagonal line is $\pis(x,z) = - \, \dfrac{\mu}{\sig^2} \, \dfrac{U_w(\ws)}{\ws U_{ww}(\ws)} \, x$, which equals the optimal investment policy when $x > \ws z$. The vertical dotted lines correspond to $x=\wa=3.703$, $x=\wo=5.5947$, and $x=\ws=11.2992$, respectively. The parameters are as in Figure \ref{fig:EtaY_star}. Note that changing $z$ would simply scale the graph. Furthermore, the graphs clearly show that $\pis(x,z)$ satisfies Conditions (b) and (c) of Corollary \ref{cor:prop}.
}
\label{fig:piSTAR}
\end{figure}

\begin{figure}[t!]
\centerline{
% \fbox{
\adjustbox{trim={0.04\width} {0.03\height} {0.08\width} {0.55\height},clip}{\includegraphics[scale=.4,page=1]{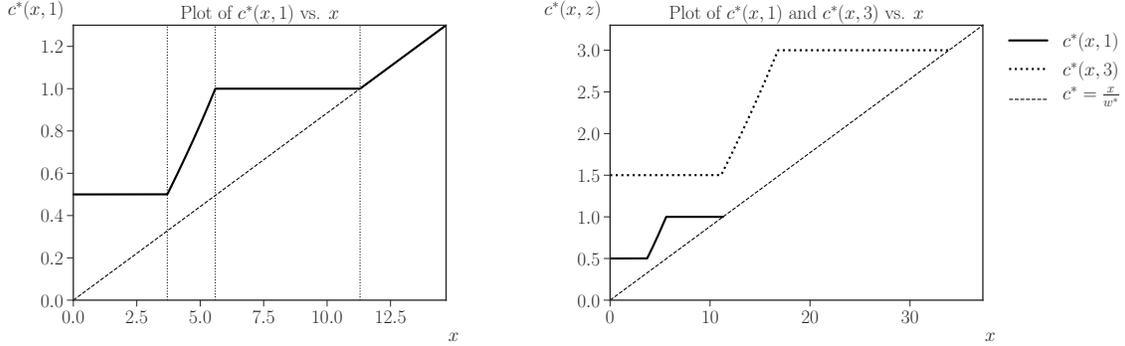}}
% }
}
\caption{Graph of $\cs(x, 1)$ on the left, and graphs of $\cs(x,1)$ and $\cs(x,2)$ on the right. In both graphs, the dashed diagonal line is $\cs(x,z) = \dfrac{x}{\ws}$, which equals the optimal dividend policy when $x > \ws z$. The vertical dotted lines correspond to $x=\wa=3.703$, $x=\wo=5.5947$, and $x=\ws=11.2992$, respectively. The parameters are as in Figure \ref{fig:EtaY_star}. Note that changing $z$ would simply scale the graphs. Furthermore, the graphs clearly show that $\cs(x,z)$ satisfies Condition (a) of Corollary \ref{cor:prop}. Finally, despite appearances, $c^*(x,z)$ is not linear in $x$ over the range $x \in (\wa z, \wo z)$. The reason that it appears almost linear is that, in equation \eqref{eq:y-sol}, the term involving $\big(y(x/z)\big)^{-1/p}$ dominates the term involving $y^{-(1 + \kap \del)}$.
% See the next figure.
% \comment{Bahman: $\cs(x,1)$ seems to be linear in x for $\wa<x<\wo$, regardless of the choice of parameters. I cannot establish if this is due to small curvature, or the graph is really a line. Any comment?}
}
\label{fig:cSTAR}
\end{figure}

For $0 < x < \ws z$, $\pis$ is proportional to $y \Uh_{yy}$; thus, except at $x = \wa z$ and $x = \wo z$,
\begin{equation}
\dfrac{\partial \pis}{\partial x} \propto \( \Uh_{yy} + y \Uh_{yyy} \) \dfrac{\partial y}{\partial x} \propto - \, \dfrac{\Uh_{yy} + y \Uh_{yyy}}{\Uh_{yy}}.
\end{equation}
For $0 < x < \wa z$, from the expression for $\Uh$ in \eqref{eq:Uh}, we obtain
\begin{equation}
- \( \Uh_{yy} + y \Uh_{yyy} \) \propto C_2 > 0.
\end{equation}
For $\wa z < x < \wo z$, or equivalently, $1 < y < \al^{-p}$,
\begin{align}
- \( \Uh_{yy} + y \Uh_{yyy} \) &\propto \kap \del (1 + \kap \del)^2 C_4 y^{-(1 + \kap \del)} + \dfrac{\kap}{p(1 + \kap \del) - 1} \, y^{- \, \frac{1}{p}} \\
&\propto (1 + \kap \del) \left\{ - \, \dfrac{1}{1 + \kap \del} + (\ys)^{\kap \del} - \dfrac{\kap \del}{1 + \kap \del} (\ys)^{1 + \kap \del} \right\} y^{-(1 + \kap \del)} + y^{- \, \frac{1}{p}} \\
&= \( y^{- \, \frac{1}{p}} - y^{-(1 + \kap \del)} \) + \( \dfrac{1 + \kap \del}{\ys} - \kap \del \) \( \dfrac{\ys}{y} \)^{1 + \kap \del} > 0,
\end{align}
in which the second line follows from the expression for $C_4$ in \eqref{eq:C4_ys}.  In the last line, the expression in the first set of parentheses is positive because $y > 1$ and $p(1 + \kap \del) > 1$, and the expression in the second set of parentheses is positive because $0 < \ys < 1$.  Finally, for $\wo z < x < \ws z$,
\begin{equation}
- \( \Uh_{yy} + y \Uh_{yyy} \) \propto C_6 > 0,
\end{equation}
We have shown that $\pis$ increases with $x$.

For $0 < x < \ws z$, except at $x = \wa z$ and $x = \wo z$,
\begin{equation}
\dfrac{\partial^2 \pis}{\partial x^2} \propto - \, \dfrac{\partial}{\partial x} \dfrac{y \Uh_{yyy}}{\Uh_{yy}} \propto - \left\{ \Uh_{yy} \( \Uh_{yyy} + y \Uh_{yyyy} \) - y \Uh^2_{yyy} \right\}\dfrac{\partial y}{\partial x} \propto \Uh_{yy} \( \Uh_{yyy} + y \Uh_{yyyy} \) - y \Uh^2_{yyy}.
\end{equation}
For $0 < x < \wa z$, from the expression for $\Uh$ in \eqref{eq:Uh}, we obtain
\begin{equation}
\Uh_{yy} \( \Uh_{yyy} + y \Uh_{yyyy} \) - y \Uh^2_{yyy} \propto C_2 > 0,
\end{equation}
because $C_2 > 0$.  We obtain a similar expression for $\wo z < x < \ws z$, except with $\al$ and $C_2$ replaced by $1$ and $C_6 > 0$, respectively.  Thus, $\pis$ is convex in $x$ for $0 < x < \wa z$ and $\wo z < x < \ws z$.  For $\wa z < x < \wo z$,
\begin{equation}
\Uh_{yy} \( \Uh_{yyy} + y \Uh_{yyyy} \) - y \Uh^2_{yyy} \propto C_4 < 0;
\end{equation}
thus, $\pis$ is concave in $x$ for $\wa z < x < \wo z$.
\end{proof}

\begin{remark}
It is interesting that $\pis$ is convex in $x$ when the constraints $\al z \le c \le z$ bind in $\Dc$.  It's as if the company invests more aggressively with increasing surplus so that surplus will further increase and the company can either avoid the lower constraint $\al z$ or increase its dividend rate beyond the upper constraint $z$ via singular control along the free-boundary $x = \ws z$.  \qed
\end{remark}

Next, we examine the effect of changing the drawdown parameter $\al$ on the value function and the optimal policy. Figure \ref{fig:wVSal} illustrates the sensitivity of the free boundaries $\wa$, $\wo$, and $\ws$ to $\al$.  The graph indicates that all three are are increasing with respect to $\al$.  In the next corollary, we prove the result for $\ws$.

\begin{figure}[b!]
\vspace{10pt}
\centerline{
% \fbox{
\adjustbox{trim={0\width} {0.05\height} {0.0\width} {0.2\height},clip}{\includegraphics[scale=.4,page=1]{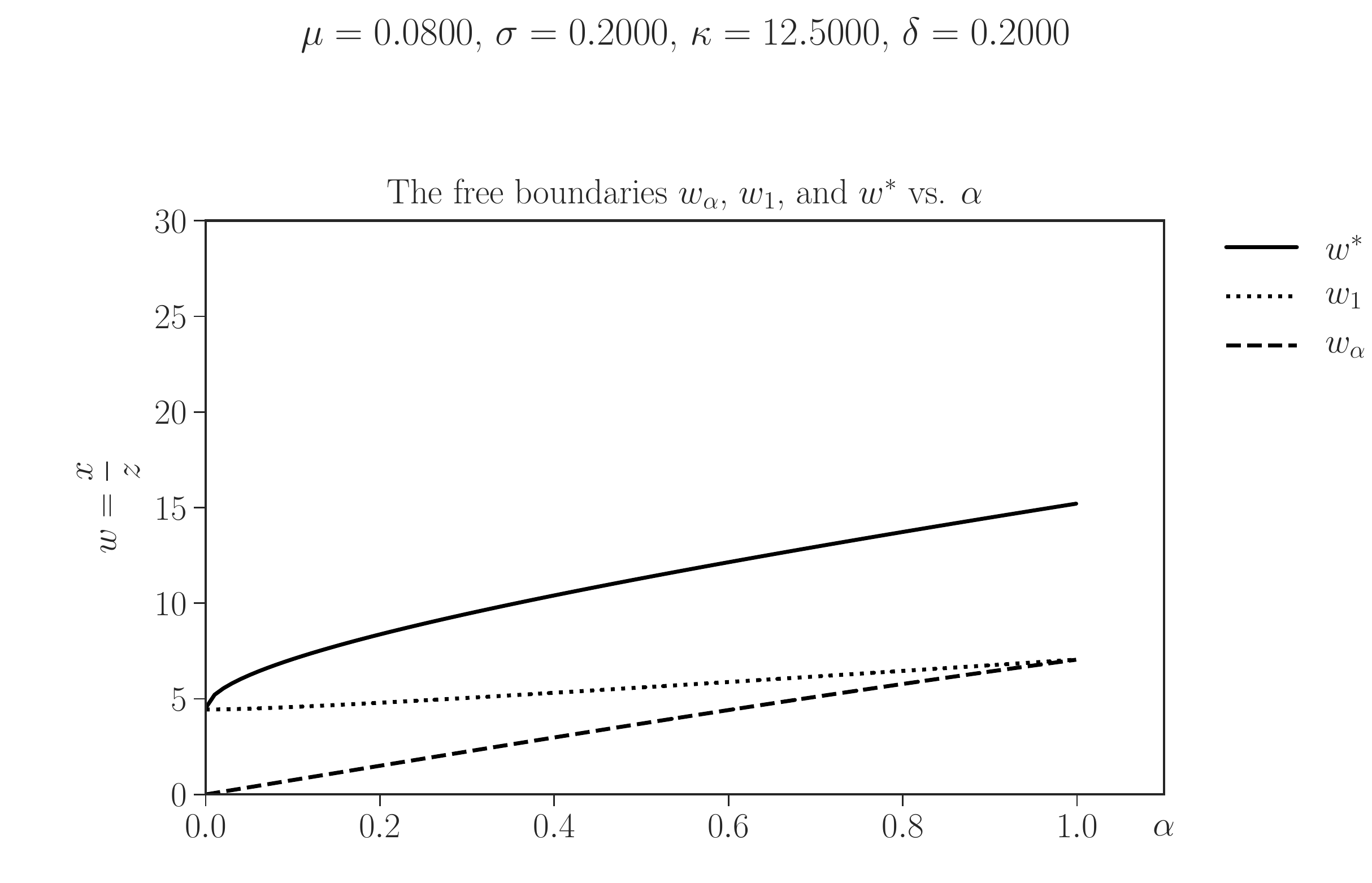}}
% }
}
\caption{Sensitivity of the free boundaries $\wa$, $\wo$, and $\ws$ with respect to $\al$.  As proved in Corollary \ref{cor:ws_al}, $\ws$ is increasing in $\al$. It appears that the other two boundaries are also increasing in $\al$.}
\label{fig:wVSal}
\end{figure}

% \pagebreak

\begin{corollary}\label{cor:ws_al}
The free boundary $\ws$ given in \eqref{eq:ws} increases with respect to $\al$.
\end{corollary}

\begin{proof}
By differentiating the expression in \eqref{eq:ws} with respect to $\al$, we obtain
\begin{equation}
\dfrac{\partial \ws}{\partial \al} = - \, \dfrac{\kap p}{p(1 + \kap \del) -1} \, \dfrac{1}{(\ys)^2} \, \dfrac{\partial \ys}{\partial \al} \propto - \, \dfrac{\partial \ys}{\partial \al}.
\end{equation}
Thus, if we show that $\ys$ decreases with respect to $\al$, then we are done.  To that end, differentiate the first equation in \eqref{eq:etas_ys} with respect to $\al$ to obtain
\begin{equation}\label{eq:etas_diff}
\( \dfrac{\al}{\etas} - \dfrac{\al}{(\etas)^2(1 - p)} \) \dfrac{\partial \etas}{\partial \al} = \( \dfrac{1}{\ys} - \dfrac{1}{(\ys)^2} \) \dfrac{\partial \ys}{\partial \al} - \dfrac{1}{\al} \( \ln \ys + \dfrac{1}{\ys} \).
\end{equation}
Next, differentiate the second equation in \eqref{eq:etas_ys} with respect to $\al$ and rearrange the resulting equation to obtain
\begin{align}
&(\etas)^{1 + \kap \del} \( \dfrac{\al}{\etas} - \dfrac{\al}{(\etas)^2(1 - p)} \) \dfrac{\partial \etas}{\partial \al} + \dfrac{\al^{p(1 + \kap\del) - 1}}{\kap \del} \( \dfrac{\kap \del}{1 + \kap \del} (\ys)^{1 + \kap \del} - (\ys)^{\kap \del} \) \\
&\quad + \dfrac{\al^{p(1 + \kap\del) - 1}}{p(1 + \kap \del) - 1} (\ys)^{1 + \kap \del} \( \dfrac{1}{\ys} - \dfrac{1}{(\ys)^2} \) \dfrac{\partial \ys}{\partial \al} = - \, \dfrac{\al^{p(1 + \kap\del) - 1}}{\kap \del(1 + \kap \del)} \, .
\end{align}
Into the above expression, substitute for the left side of \eqref{eq:etas_diff} and rearrange to get
\begin{equation}
\dfrac{\partial \ys}{\partial \al} \propto \dfrac{\al^{p(1 + \kap\del) - 1}}{\kap \del} \(  \dfrac{1}{1 + \kap \del} + \dfrac{\kap \del}{1 + \kap \del} (\ys)^{1 + \kap \del} - (\ys)^{\kap \del} \) - \dfrac{(\etas)^{1 + \kap \del}}{\al} \( \ln \ys + \dfrac{1}{\ys} \).
\end{equation}
Use the equations in \eqref{eq:etas_ys} to write this expression in terms of $\etas$; specifically,
\begin{align}
\dfrac{\partial \ys}{\partial \al} &\propto \dfrac{1}{\kap \del(1 + \kap \del)} - \dfrac{p(1 + \kap \del) - 1}{\kap \del} \( \dfrac{\kap \del}{1 + \kap \del} (\etas)^{1 + \kap \del} - \dfrac{1}{1 - p} (\etas)^{\kap \del} \) \\
& \quad - \dfrac{(\etas)^{1 + \kap \del}}{\al} \left\{ \al \( \ln \etas + \dfrac{1}{\etas(1 - p)} \) - \al(1 + p) + 1 \right\} \\
&\propto \dfrac{(\etas)^{-(1 + \kap \del)}}{\kap \del(1 + \kap \del)} + \( \dfrac{1}{1 + \kap \del} + 1 - \dfrac{1}{\al} \) - \dfrac{1 + \kap \del}{\kap \del \etas} - \ln \etas. 
\end{align}
Define $f$ by
\begin{equation}
f(\eta) = \dfrac{\eta^{-(1 + \kap \del)}}{\kap \del(1 + \kap \del)} + \( \dfrac{1}{1 + \kap \del} + 1 - \dfrac{1}{\al} \) - \dfrac{1 + \kap \del}{\kap \del \eta} - \ln \eta.
\end{equation}
for $\eta \ge 1$.  Note that $f(1) = - 1/\al < 0$, and
\begin{equation}
f'(\eta) \propto \dfrac{1 + \kap \del}{\eta} - \dfrac{1}{\eta^{1 + \kap \del}} - \kap \del =: g(\eta).
\end{equation}
Note that $g(1) = 0$ and
\begin{equation}
g'(\eta) \propto 1 - \eta^{\kap \del} < 0, \qquad \hbox{for } \eta > 1.
\end{equation}
Thus, $g(\eta) \le 0$ for $\eta \ge 1$, which implies that $f'(\eta) \le 0$ for $\eta \ge 1$, which further implies that $f(\eta) < 0$ for $\eta \ge 1$.  Therefore, because $f$ is proportional to $\frac{\partial \ys}{\partial \al}$, we have shown that $\ys$ decreases with $\al$, and $\ws$ increases with $\al$.
\end{proof}

\begin{remark}
It is intuitively pleasing that $\ws$ increases with the drawdown parameter $\al$ because as $\al$ increases, the lower bound on the excess dividend rate increases, and we expect the company to be less willing to increase its historical peak.   \qed
\end{remark}

Figures \ref{fig:VF_al}, \ref{fig:piSTAR_al}, and \ref{fig:cSTAR_al} illustrate, respectively, how the value function $V(x,1)$, the optimal investment policy $\pis(x,1)$, and the optimal dividend policy $\cs(x,1)$ change with respect to $\al\in(0,1)$.

In the final two corollaries, we consider limiting cases of our problem, namely, $\al \to 0+$ and $\al \to 1-$, respectively.  Without working through all the details (that is, providing an explicit verification theorem and proving that the proposed solution is the classical solution of the resulting variational inequality), one can show that the limit of the expression in \eqref{eq:V-sol} is the solution of the limiting problem.  In other words, our solution is continuous with respect to the drawdown parameter $\al$.

\begin{figure}[t!]
\centerline{
% \fbox{
\adjustbox{trim={0.08\width} {0.\height} {0.03\width} {0.15\height},clip}{\includegraphics[scale=.35,page=1]{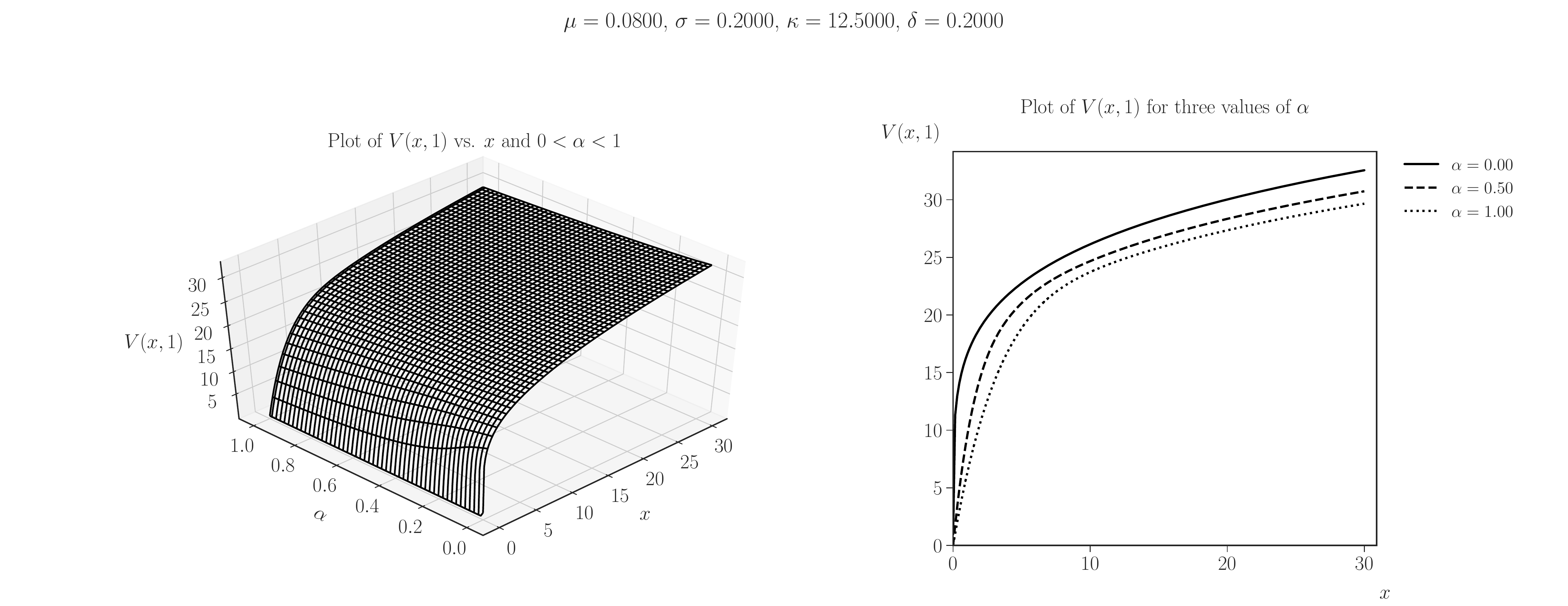}}
% }
}
\caption{On the left: the value function $V(x,1)$ for different values of $\al$. As expected, increasing $\al$ decreases the value function, since this reduces the set of admissible policies. On the right, three value functions are plotted, namely, the value function of the Merton problem $\al=0$ (the solid curve), the value function for the ratcheting problem $\al=1$ (the dotted curve), and the value function for a drawdown problem with $\al=0.5$ (the dashed curve). 
}
\label{fig:VF_al}
\end{figure}
\begin{figure}[t!]
\centerline{
% \fbox{
\adjustbox{trim={0.08\width} {0.\height} {0.03\width} {0.15\height},clip}{\includegraphics[scale=.35,page=1]{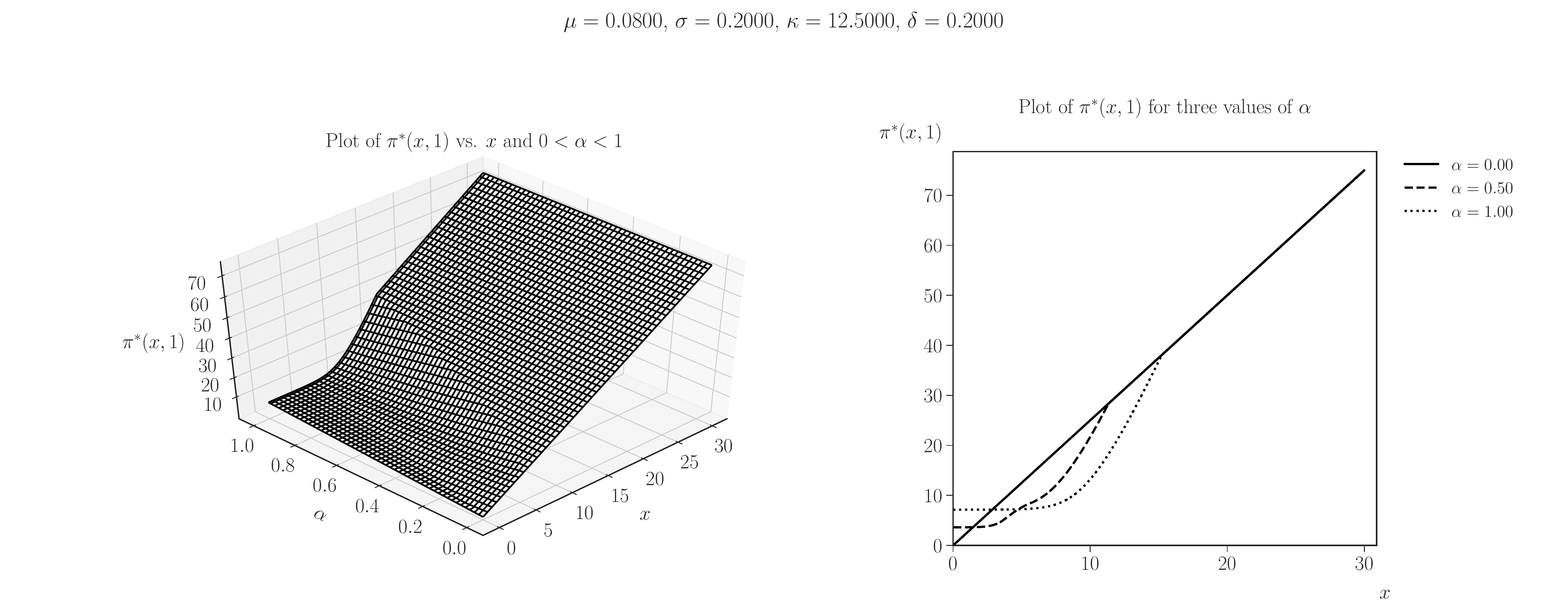}}
% }
}
\caption{On the left: the optimal investment policy $\pis(x,1)$ for different values of $\al$. On the right, the optimal policy $\pis(x,1)$ for the Merton problem (solid curve with $\al=0$), the ratcheting problem (dotted curve for $\al=1$), and a drawdown problem with $\al=0.5$ (the dashed curve). 
}
\label{fig:piSTAR_al}
\end{figure}

\begin{figure}[b!]
\centerline{
% \fbox{
\adjustbox{trim={0.08\width} {0.\height} {0.03\width} {0.15\height},clip}{\includegraphics[scale=.35,page=1]{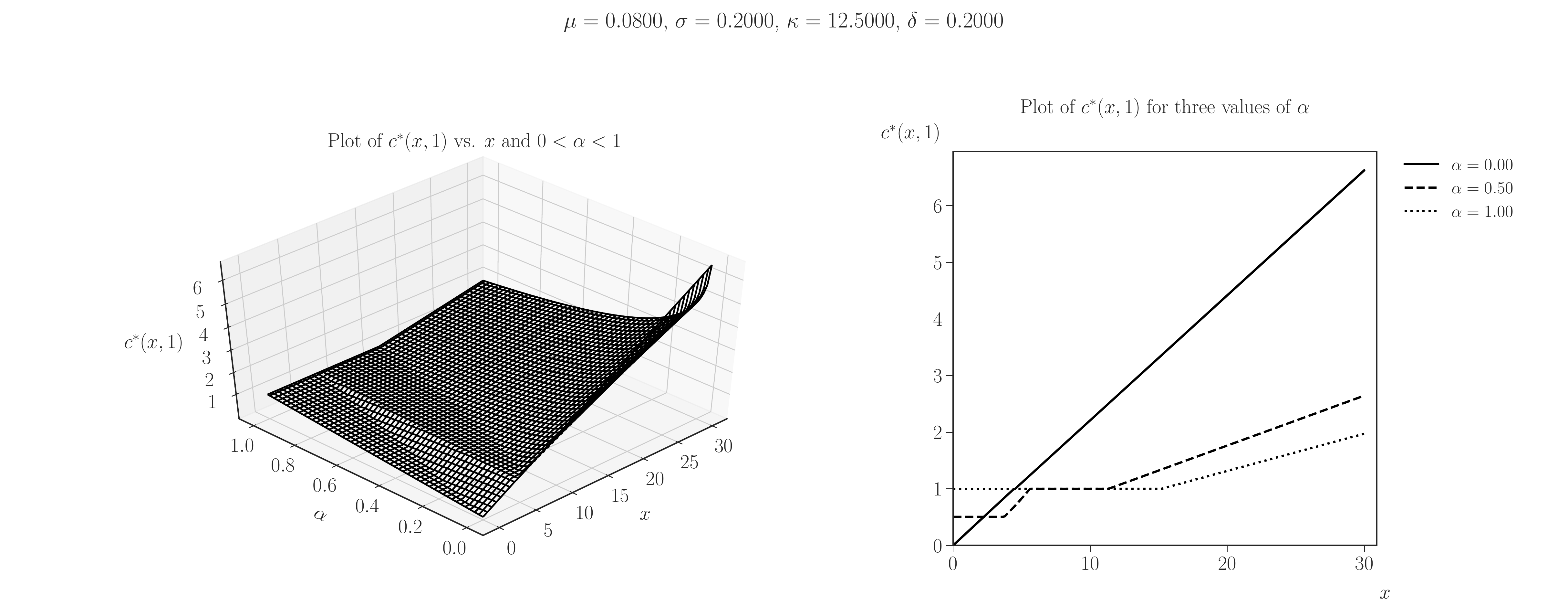}}
% }
}
\caption{On the left: the optimal dividend policy $\cs(x,1)$ for different values of $\al$. On the right, the optimal dividend policy $\cs(x,1)$ for the Merton problem (solid curve with $\al=0$), the ratcheting problem (dotted curve for $\al=1$), and the a drawdown problem with $\al=0.5$ (the dashed curve). 
}
\label{fig:cSTAR_al}
\end{figure}

In the first case, $V$ in \eqref{eq:V} becomes the value function for the following optimization problem:
\begin{equation}\label{eq:V0}
V_0(x) = \sup_{(\pi_t, c_t) \in \Cb_0} \E^x \[ \int_0^\infty \ee^{-\del t} \, \frac{c_t^{1-p}}{1-p} \, \dd t \],
\end{equation}
which is the infinite-horizon, optimal-consumption Merton problem; see Section $VI$ in \cite{Merton1969}.  (Under the optimal policy, $\tau = \infty;$ thus, we may replace $\tau$ with $\infty$ in the integral.)  Recall that $\Cb_0$ is the set of strategies with no drawdown constraint on the excess dividend rate, that is, we only require $c_t \ge 0$.  Thus, the set of admissible strategies $\Cb_0 = \Cb(0,z)$ is independent of the value of $z$, and the optimal policy does not depend on $\{z_t\}$, as we observe below in \eqref{eq:pi0} and \eqref{eq:c0}.

\begin{corollary}\label{cor:al0}
If we let $\al \to 0+$, then \eqref{eq:V-sol} becomes
	\begin{align}\label{eq:V0-sol}
		V_0(x) = \frac{x^{1-p}}{1-p} \(\frac{\kap p^2}{p(1 + \kap \del) - 1}\)^p ,
	\end{align}
for $x > 0$, and the optimal investment and excess dividend feedback functions are given by
\begin{equation}\label{eq:pi0}
\pib(x) = \frac{\mu}{\sig^2 p} \, x,
\end{equation}
and
\begin{equation}\label{eq:c0}
\cb(x) = \frac{p (1 + \kap \del) - 1}{\kap p^2} \, x.
\end{equation}
\end{corollary}

\smallskip

The second limiting case, $\al \to 1-$, means that we impose a \textit{ratcheting} constraint on the excess dividend rate, that is, the excess dividend rate is never allowed to decrease below its current level, as in \cite{Dybvig1995}.  However, \cite{Dybvig1995} frames his problem in such a way that the probability of ruin was $0$, that is, he required that the dividend (or consumption) rate $C_t \le rX_t$.  By contrast, we require that $C_t \ge rX_t$.

% \mcomment{Bahman: Plot $V(x,1)$, $\pis(x,1)$, and $\cs(x,1)$ for $\al=0,0.5, 1$.}

\begin{corollary}\label{cor:al1}
If we let $\al \to 1-$, then \eqref{eq:V-sol} becomes
\begin{equation}\label{eq:V1-sol}
V(x, z) = 
\begin{cases}
\dfrac{z^{1-p}}{p(1 + \kap \del) - 1} \( \dfrac{1}{\del} - \dfrac{\kap \ys}{1 + \kap \del} \) \( \dfrac{\ys}{y} \)^{\kap \del} + \dfrac{z^{1 - p}}{\del(1 - p)} - \dfrac{\kap z^{1 - p} y}{1 + \kap \del} \, , &\quad 0 \le x \le \ws z, \vspace{2ex} \\
\dfrac{x^{1 - p}}{1 - p} \, \dfrac{\kap p}{p(1 + \kap \del) - 1} \, \dfrac{(\ws)^p}{\ws + \kap p(1 - p)} \, , &\quad x > \ws z,
\end{cases}
\end{equation}
and the optimal investment and excess dividend feedback functions are given by
\begin{equation}\label{eq:pi1}
\pi_1(x, z) = 
\begin{cases}
\dfrac{2z}{\mu} \left\{ \dfrac{1}{1 + \kap \del} + \dfrac{1}{p(1 + \kap \del) - 1} \( \dfrac{1}{\ys} - \dfrac{\kap \del}{1 + \kap \del} \) \( \dfrac{\ys}{y} \)^{1 + \kap \del} \right\}, &\quad 0 \le x \le \ws z, \vspace{2ex} \\
\dfrac{2x}{\mu} \, \dfrac{1}{\kap p \big(p(1 + \kap \del) - 1 \big)} \, ,  &\quad x > \ws z,
\end{cases}
\end{equation}
and
\begin{equation}\label{eq:c1}
c_1(x, z) = 
\begin{cases}
z,  &\quad  0 \le x \le \ws z, \vspace{1ex} \\
\dfrac{x}{\ws} \, , &\quad x > \ws z.
\end{cases}
\end{equation}
In the above expressions, given $0 \le x \le \ws z$, $y \in [\ys, y_0]$ uniquely solves
\begin{equation}
\dfrac{x}{z} = \dfrac{\kap}{1 + \kap \del} \left\{ \ln \dfrac{\ys}{y} + p + \( \dfrac{1}{\ys} - \dfrac{\kap \del}{1 + \kap \del} \) \( 1 + \dfrac{ \( \ys/y \)^{1 + \kap \del} }{p(1 + \kap \del) - 1}\) \right\}.
\end{equation}
\end{corollary}

\smallskip

\begin{remark}
In these two corollaries, we see that the drawdown parameter $\al$ provides a link between Merton's optimal consumption problem and the ratcheting problem.  \cite{Arun2012} observes the same connection between Merton's optimal consumption problem and the ratcheting problem of \cite{Dybvig1995}.  Recall that neither the results of \cite{Arun2012} nor of \cite{Dybvig1995} are comparable to ours because they both require the dividend/consumption rate to be such that ruin is impossible, while our model allows ruin to occur.  \qed
\end{remark}

We end the paper by commenting on the effect of the risk aversion parameter on the value function. As $p\to\frac{1}{1+\kap\,\del}$, the value function become arbitrary large and the control problem is ill-posed for $0\le p\le \frac{1}{1+\kap\,\del}$. As pointed out earlier, this phenomenon has been observed in other studies such as \cite{Merton1969}. Furthermore, since the utility function $\frac{c^{1-p}}{1-p}$ goes to infinity as $p\to 1^-$, the value function explodes there, too. Figure \ref{fig:VF_p} illustrates these properties.

\begin{figure}[t!]
\centerline{
% \fbox{
\adjustbox{trim={0.08\width} {0.\height} {0.03\width} {0.15\height},clip}{\includegraphics[scale=.35,page=1]{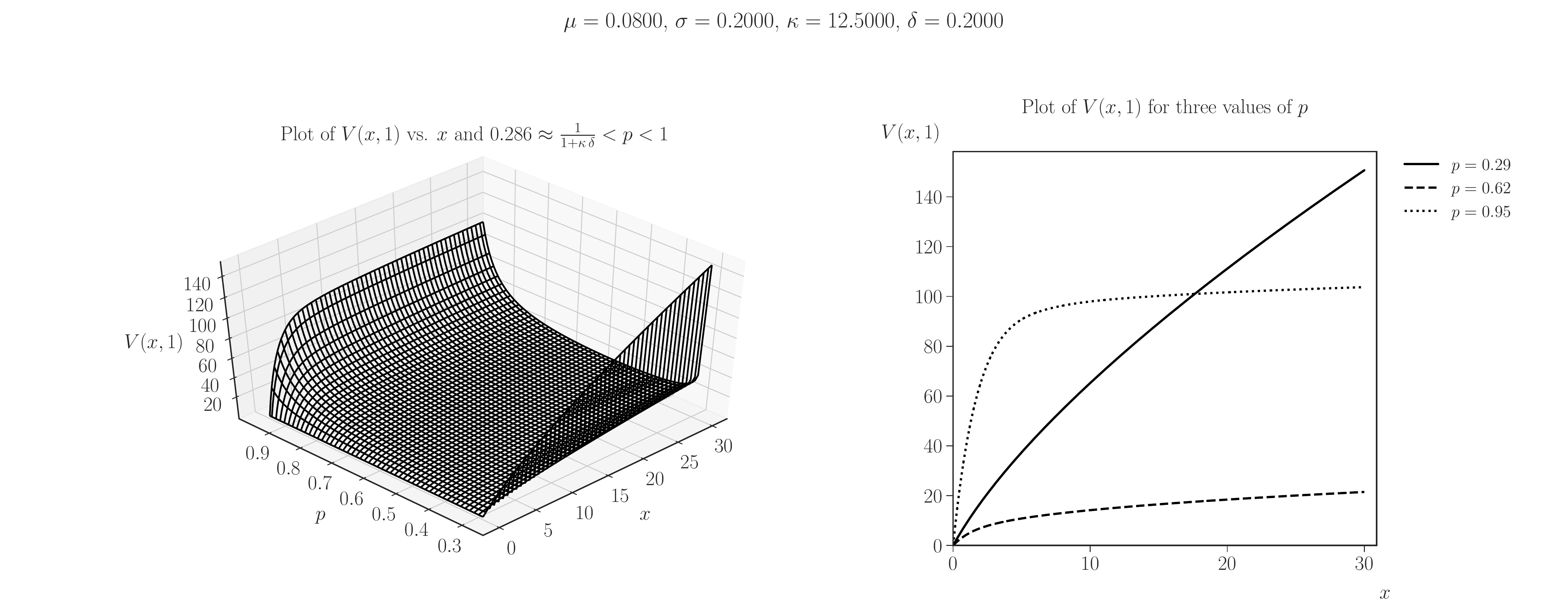}}
% }
}
\caption{
On left: the value function $V(x,1)$ for different values of $p$ between $\frac{1}{1+\kap\,\del}\approx0.286$ and 1.
Note that as $p\to1^-$, the utility function $\frac{c^{1-p}}{1-p}\to\infty$ for all $c$.
Thus, the value function goes to infinity as $p\to1^-$. Furthermore, as in \cite{Merton1969}, the value function also explodes as $p\to \frac{1}{1+\kap\,\delta}$. On the right, the value function $V(x,1)$ is plotted for three values of $p$, namely, $p=0.29$ (the solid curve), $p=0.62$ (the dotted curve), and $p=0.95$ (the dashed curve).
}
\label{fig:VF_p}
\end{figure}

\appendix

\section{Proof of Theorem \ref{thm:V}}\label{app:Verify}

For the clarity of the arguments in this section, we use the notation $\Vt$ to refer to the value function given by \eqref{eq:V}, that is,
\begin{equation}
	\Vt(x, z) := \sup_{(\pi_t, c_t) \in \Cb(\al, z)} \E^x \[ \int_0^\tau \ee^{-\del t} \, \frac{c_t^{1-p}}{1-p} \, \dd t \];\quad (x,z)\in\Rb^2_+.
\end{equation}
We reserve the notation $V$ to refer to the solution of the free-boundary problem \eqref{eq:HJB-FB}, given by \eqref{eq:V-sol} on $\Dc$ and \eqref{eq:V_xbig} on $\Rb_+^2-\Dc$.

The proof is dividend into two parts. For all $x,z\ge0$, we show in Section \ref{app:UBound} that $V(x,z)\ge\Vt(x,z)$. Then, in Section \ref{app:LBound}, we show that $V(x,z)\le\Vt(x,z)$.  We, thereby, prove that $V$ is the value function.

\subsection{Showing that $V(x,z)\ge\Vt(x,z)$}\label{app:UBound}

We start by providing a (smooth) comparison lemma for \eqref{eq:V}.

\begin{lemma}\label{lem:verif}
Suppose that function $v: \Rb^2_+ \to \Rb$ is continuously twice differentiable in $x$ and continuously differentiable in $z$, such that, for all $x ,z \ge 0$, $\pi\in\Rb$, and $c \ge \al z$,
	\begin{enumerate}
		\item[$(i)$] $v_z(x,z) \le 0$, and
		\item[$(ii)$] $\dfrac{1}{2} \, \sig^2 \pi^2 v_{xx}(x,z) + (\mu \pi - c) v_x(x,z) - \del v(x,z) + \dfrac{c^{1-p}}{1-p} \le 0$.
	\end{enumerate}
	Then, $v(x, z) \ge \Vt(x, z)$, for all $x, z \ge0$.
\end{lemma}

\begin{proof}
	First, note that by Condition (ii),
	\begin{align}\label{eq:LB}
		v(x,z)\ge 0;\quad \forall x,z\ge0.
	\end{align}
	To see this, let $\pi=0$, so that Condition (ii) becomes
	\begin{align}
		- c v_x(x,z) - \del v(x,z) + \dfrac{c^{1-p}}{1-p} \le 0.
	\end{align}
	Note that $v_x(x,z)\ge 0$, otherwise this equation is violated by choosing $c>-\del v(x,z) / v_x(x,z)$. By maximizing over all $c\ge0$, we then obtain
	\begin{align}
		\del v(x,z) \ge \sup_{c\ge0}\left\{\dfrac{c^{1-p}}{1-p} - c\, v_x(x,z)\right\} = \frac{p}{1-p} \big(v_x(x,z)\big)^{\frac{p-1}{p}}\ge0,
	\end{align}
	which, in turn, yields $v(x,z)\ge0$.
	 
Fix arbitrary initial values $x, z \ge 0$, and choose a policy $(\pi_t, c_t) \in  \Cb(\al, z)$. Let the processes $(z_t)$ and $(X_t)$ be given by \eqref{eq:z} and \eqref{eq:X}, respectively. Define $\tau_n$ by
	\begin{equation}
		\tau_n = n \wedge \tau \wedge \inf \left\{t > 0: \int_0^t \ee^{-\del s} \pi_s^2 v_x^2(X_s, z_s) \dd s \ge n \right\}.
	\end{equation}
Recall that $\tau$ depends on the policy $(\pi_t, c_t)$. Applying It\^o's lemma to $\ee^{-\del t} \, v(X_t,z_t)$ on $[0,\tau_n]$ yields
\begin{align*}
	\ee^{-\del {\tau_n}} \, v(X_{\tau_n}, z_{\tau_n}) &= v(x, z) - \int_0^{\tau_n} \ee^{-\del t} \, \frac{c_t^{1-p}}{1-p} \, dt \\
	&\quad + \int_0^{\tau_n} \ee^{-\del t} \Bigg[ \frac{1}{2} \sig^2 \pi^2_t v_{xx} + (\mu \pi_t - c_t) v_x - \del v + \frac{c_t^{1-p}}{1-p}\Bigg] dt 
	+ \int_0^{\tau_n} \sig \, \ee^{-\del t} \pi_t \, v_x \, dW_t \\
	&\quad + \int_0^{\tau_n} \ee^{-\del t} \, v_z \, dz_t^c + \sum_{0 \le t \le {\tau_n}} \ee^{-\del t}\Big(v(X_t,z_{t^+}) - v(X_t,z_t)
	% - v_z(X_t,z_t) \Delta z_t
	\Big),
\end{align*}
in which $(z_t^c)$ is the continuous version of $(z_t)$ obtained by removing the jumps. Taking expectations of both sides yields
\begin{equation}\label{eq:verif1}
\begin{split}
	\E^x \int_0^{\tau_n} \ee^{-\del t} \, \frac{c_t^{1-p}}{1-p} \, \dd t &= 
	v(x,z) - \Eb^x\Big(\ee^{-\del {\tau_n}} \, v(X_{\tau_n}, z_{\tau_n})\Big)\\
	&\quad+ \E^x \int_0^{\tau_n} \ee^{-\del t} \Bigg[ \frac{1}{2} \sig^2 \pi^2_t v_{xx} + (\mu \pi_t - c_t) v_x - \del v + \frac{c_t^{1-p}}{1-p}\Bigg] \dd t\\
	&\quad+ \E^x \bigg(\int_0^{\tau_n} \ee^{-\del t} \, v_z \, \dd z_t^c + \sum_{0 \le t \le {\tau_n}} \ee^{-\del t} \Big(v(X_t,z_{t^+})-v(X_t,z_t) 
	% - v_z(X_t,z_t) \Delta z_t
	\Big) \bigg).
\end{split}
\end{equation}
By \eqref{eq:LB} and Conditions (i) and (ii), the terms on the right side involving expectations are non-positive; therefore,
\begin{equation}
	v(x,z) \ge
	 \E^x \int_0^{\tau_n} \ee^{-\del t} \, \frac{c_t^{1-p}}{1-p} \, \dd t.
\end{equation}
By letting $n\to \infty$, we obtain
\begin{equation}\label{eq:verif2}
	v(x,z) \ge \E^x \int_0^\tau \ee^{-\del t} \, \frac{c_t^{1-p}}{1-p} \, \dd t.
\end{equation}
The inequality $v \ge \Vt$ follows by taking the supremum over all $(\pi_t, c_t) \in \Cb(\al, z)$.		
\end{proof}

Next, we show that $V(x, z)\ge\Vt(x,z)$ for all $x, z \ge0$ by checking that $V$ satisfies the conditions in Lemma \ref{lem:verif}.

The differentiability conditions are readily verified from \eqref{eq:V-sol} and \eqref{eq:V_xbig}. Below, we check Conditions $(i)$ and $(ii)$ of Lemma \ref{lem:verif}.
\medskip

\noindent\textbf{Condition $(i)$:} On $\Rb_+^2 - \Dc$, from \eqref{eq:V_xbig} it follows that $V_z = 0$. It only remains to show that $V_z\le0$ on $\Dc$. Because $V(x, z) = z^{1 - p} U(x/z)$, inequality $V_z \le 0$ is equivalent to $(1 - p)U - w U_w \le 0$.  By rewriting this inequality in terms of $\Uh$, we deduce that $V_z \le 0$ is equivalent to $(1 - p) \Uh + p y \Uh_y \le 0$.  By construction, we know that $(1- p) \Uh(\ys) + p \ys \Uh_y(\ys) = 0$ and $(1 - p) \Uh(y_0) + p y_0 \Uh_y(y_0) = 0$. Thus, if we show that $(1 - p) \Uh + p y \Uh_y$ decreases then increases as $y$ increases from $\ys$ to $y_0$, then we will have shown that $V_z \le 0$ on $\Dc$.

Begin by considering $(1 - p) \Uh + p y \Uh_y$ on the interval $\ys \le y \le 1$; then, $(1 - p) \Uh + p y \Uh_y$ decreases if and only if
\begin{equation}
\dfrac{d}{dy} \( (1 - p) \Uh(y) + p y \Uh_y(y) \) = \Uh_y(y) + py \Uh_{yy}(y) \le 0,
\end{equation}
which, on this interval, is equivalent to
\begin{equation}\label{ineq:Vz_1}
C_5 + \kap \del \big( p(1 + \kap \del) - 1 \big) C_6 y^{-(1 + \kap \del)} + \dfrac{\kap}{1 + \kap \del} \big( 1 + p + \ln y \big) \le 0.
\end{equation}
By substituting the expressions for $C_5$ and $C_6$ given in \eqref{eq:C5_ys} and \eqref{eq:C6_ys}, respectively, and by simplifying the result, inequality \eqref{ineq:Vz_1} becomes
\begin{equation}
\( \dfrac{1}{\ys} - \dfrac{\kap \del}{1 + \kap \del} \) \( 1 - \( \dfrac{\ys}{y} \)^{1 + \kap \del} \) + \ln \dfrac{\ys}{y} \ge 0.
\end{equation}
Define $f$ by
\begin{equation}
f(x) = \( \dfrac{1}{\ys} - \dfrac{\kap \del}{1 + \kap \del} \) \( 1 - x^{1 + \kap \del} \) + \ln x,
\end{equation}
for $\ys \le x \le 1$.  First, show that $f(\ys) > 0$ for \textit{any} $\ys \in (0, 1)$. To that end, define $g$ by
\begin{equation}
g(x) = \( \dfrac{1}{x} - \dfrac{\kap \del}{1 + \kap \del} \) \( 1 - x^{1 + \kap \del} \) + \ln x,
\end{equation}
for $0 < x \le 1$.  It is easy to show that $\lim_{x \to 0+} g(x) = + \infty$ because $1/x$ dominates $\ln x$, $g(1) = 0$, and $g'(x) < 0$ for $0 < x < 1$.  Thus, $g(x) > 0$ for all $0 < x < 1$, which implies that $f(\ys) > 0$.  Next, observe that $f(1) = 0$, and
\begin{equation}
f''(x) = - \kap \del (1 + \kap \del) \( \dfrac{1}{\ys} - \dfrac{\kap \del}{1 + \kap \del} \) x^{\kap \del - 1} - \dfrac{1}{x^2} < 0.
\end{equation}
Because $f$ is concave, we deduce that $f(x) \ge 0$ for all $\ys \le x \le 1$.  Thus, we have shown that $(1 - p) \Uh + p y \Uh_y$ decreases on the interval $\ys \le y \le 1$, which implies that $V_z \le 0$ for $\wo z \le x \le \ws z$.

Next, consider $(1 - p) \Uh + p y \Uh_y$ on the interval $1 < y < \al^{-p}$; then, $(1 - p) \Uh + p y \Uh_y$ decreases on this interval if and only if
\begin{equation}\label{ineq:Vz_2}
C_3 + \kap \del \big( p(1 + \kap \del) - 1 \big) C_4 y^{-(1 + \kap \del)} \le 0.
\end{equation}
Rewrite $C_3$ and $C_4$ from \eqref{eq:C3} and \eqref{eq:C4} in terms of $\ys$ by using the equations in \eqref{eq:etas_ys}, that is,
\begin{equation}\label{eq:C3_ys}
C_3 = - \, \dfrac{\kap}{1 + \kap \del} \left\{ \dfrac{1}{\ys} + \ln \ys - 1 \right\} ,
\end{equation}
and
\begin{equation}\label{eq:C4_ys}
C_4 = \dfrac{1}{\del(1 + \kap \del) \big( p(1 + \kap \del) - 1 \big)} \left\{ - \, \dfrac{1}{1 + \kap \del} + (\ys)^{\kap \del} - \dfrac{\kap \del}{1 + \kap \del} (\ys)^{1 + \kap \del} \right\}.
\end{equation}
Substitute $C_3$ and $C_4$ from \eqref{eq:C3_ys} and \eqref{eq:C4_ys}, respectively, into inequality \eqref{ineq:Vz_2} and simplify the result to obtain
\begin{equation}\label{ineq:Vz_3}
\left\{ \dfrac{1}{\ys} + \ln \ys - 1 \right\} + \left\{ \dfrac{1}{1 + \kap \del} - (\ys)^{\kap \del} + \dfrac{\kap \del}{1 + \kap \del} (\ys)^{1 + \kap \del} \right\} y^{-(1 + \kap \del)} \ge 0.
\end{equation}
We will demonstrate that inequality \eqref{ineq:Vz_3} holds by showing that the expressions in each of the curly brackets is positive for \textit{any} value of $\ys \in (0, 1)$.  First, define $h$ by
\begin{equation}
h(x) = \dfrac{1}{x} + \ln x - 1,
\end{equation}
for $0 < x \le 1$.  Note that $h(0+) = + \infty$, $h(1) = 0$, and $h'(x) < 0$ for $0 < x < 1$; thus, $h(x) > 0$ for $0 < x < 1$.  Second, define $j$ by
\begin{equation}
j(x) = \dfrac{1}{1 + \kap \del} - x^{\kap \del} + \dfrac{\kap \del}{1 + \kap \del} x^{1 + \kap \del},
\end{equation}
for $0 \le x \le 1$.  Note that $j(0) = \frac{1}{1 + \kap \del} > 0$, $j(1) = 0$, and $j'(x) < 0$ for $0 < x < 1$; thus, $j(x) > 0$ for $0 < x < 1$.  Thus, we have shown that $(1 - p) \Uh + p y \Uh_y$ decreases on the interval $1 < y < \al^{-p}$, which implies that $V_z \le 0$ for $\wa z < x < \wo z$.

Finally, consider $(1 - p) \Uh + p y \Uh_y$ on the interval $\al^{-p} \le y \le y_0$.  On this interval, we will show that $(1 - p) \Uh + p y \Uh_y$ is convex (that is, this expression first decreases then increases as $y$ increases from $\al^{-p}$ to $y_0$), which is equivalent to showing that
\begin{equation}\label{ineq:Vz_4}
\dfrac{d}{dy} \( C_1 + \kap \del \big( p(1 + \kap \del) - 1 \big) C_2 y^{-(1 + \kap \del)} + \dfrac{\kap \al}{1 + \kap \del} \big( 1 + p + \ln y \big) \) > 0.
\end{equation}
After substituting the expression for $C_2$ from \eqref{eq:C2}, writing $\etas = y_0 \al^p$, and simplifying the result, inequality \eqref{ineq:Vz_4} becomes
\begin{equation}
\dfrac{1}{1 + \kap \del} - \big( p(1 + \kap \del) - 1 \big) \left\{ \dfrac{\kap \del}{1 + \kap \del} - \dfrac{1}{\etas(1 - p)} \right\} \( \dfrac{y_0}{y} \)^{1 + \kap \del} > 0.
\end{equation}
Define $k$ by
\begin{equation}
k(x) = \dfrac{1}{1 + \kap \del} - \big( p(1 + \kap \del) - 1 \big) \left\{ \dfrac{\kap \del}{1 + \kap \del} - \dfrac{1}{\etas(1 - p)} \right\} x^{1 + \kap \del},
\end{equation}
for $1 \le x \le \etas$.
\begin{equation}
k'(x)= - (1 + \kap \del) \big( p(1 + \kap \del) - 1 \big) \left\{ \dfrac{\kap \del}{1 + \kap \del} - \dfrac{1}{\etas(1 - p)} \right\} x^{\kap \del}.
\end{equation}
$k'(x) < 0$ because $p(1 + \kap \del) - 1 > 0$ and the expression in the curly brackets is positive; the latter follows from the second equation in \eqref{eq:etas_ys}.  Thus, to show that $k(x) > 0$ for $1 \le x \le \etas$, it is enough to show that $k(\etas) > 0$.  After using the second equation in \eqref{eq:etas_ys} to rewrite $k(\etas) > 0$ in terms of $\ys$, we obtain
\begin{equation}
\al^{p(1 + \kap \del) - 1} \( \dfrac{1}{1 + \kap \del} - (\ys)^{\kap \del} + \dfrac{\kap \del}{1 + \kap \del} (\ys)^{1 + \kap \del} \) > 0,
\end{equation}
which we know is true from having shown that $j(x) > 0$ for all $0 < x < 1$.  Thus, we have shown that $(1 - p) \Uh + p y \Uh_y$ is convex on the interval $\al^{-p} \le y \le y_0$, which implies that $V_z \le 0$ for $0 \le x \le \wa z$.
\medskip

\noindent\textbf{Condition $(ii)$:} This condition is satisfied on $\Dc$ because, by construction, the expression for $V$ given in \eqref{eq:V-sol} satisfies the free-boundary problem in \eqref{eq:HJB-FB}. To check the condition on $\Rb_+^2 - \Dc$, we need to show that 
\begin{equation}
\mathcal{L}^{\pi, c} \, V(x, z) := \dfrac{1}{2} \, \sig^2 \pi^2 V_{xx}(x, z) + (\mu \pi - c) V_x(x, z) + \frac{c^{1-p}}{1-p}  - \del V(x, z) \le 0,
\end{equation}
for all $\pi \in \Rb$ and $c \ge \al x/\ws$, and for all $x > \ws z$. Note that on $\Rb_+^2 - \Dc$, we have obtained the expression for $V$ given in \eqref{eq:V_xbig} from $V(x, z) = V(x, x/\ws)$, in which $V(x, x/\ws)$ is given by \eqref{eq:V-sol}.  From the free-boundary conditions $V_z(x, x/\ws) = 0$ and $V_{xz}(x, x/\ws) = 0$, we deduce that, for $x > \ws z$,
\begin{equation}
\dfrac{\partial}{\partial x} V(x, z) = \dfrac{\partial}{\partial x} V(x, x/\ws) = V_x(x, x/\ws) + \dfrac{1}{\ws} V_z(x, x/\ws) = V_x(x, x/\ws),
\end{equation}
and
\begin{equation}
\dfrac{\partial^2}{\partial x^2} V(x, z) = \dfrac{\partial}{\partial x} V_x(x, x/\ws) = V_{xx}(x, x/\ws) + \dfrac{1}{\ws} V_{xz}(x, x/\ws) = V_{xx}(x, x/\ws).
\end{equation}
Thus, for $x > \ws z$, $\mathcal{L}^{\pi, c} \, V(x, z) = \mathcal{L}^{\pi, c} \, V(x, x/\ws) \le 0$ because $(x, x/\ws) \in \Dc$, and $\mathcal{L}^{\pi, c} \, V \le 0$ on $\Dc$.

\subsection{Showing that $V(x,z)\le\Vt(x,z)$}\label{app:LBound}

Let the functions $\pi^*(x, z)$ and $c^*(x,z)$ be given by \eqref{eq:pis}--\eqref{eq:cs_xbig} on $\Rb^2_+$. To show that $V(x,z)\le \Vt(x,z)$ for all $x,z\ge0$, it suffices to show that the following two conditions hold. For all $x,z\ge0$,
\begin{enumerate}
	\item[($iii$)] the following stochastic differential equation (SDE) has a unique strong solution $(X^*_t)_{t\ge0}$,
	\begin{align}\label{eq:XSTAR}
	\begin{cases}
		\dd X_t^* = \Big(\mu\,\pi^*\left(X^*_t, \frac{M^*_t}{\ws}\right) - c^*\left(X^*_t, \frac{M^*_t}{\ws}\right)\Big) \dd t + \sig \pi^*\left(X^*_t, \frac{M^*_t}{\ws}\right) \, \dd W_t;\quad t\ge0,\\
		M^*_t = \max \left\{\ws\,z, \displaystyle\sup_{0 \le s < t} X^*_s\right\};\quad t\ge0,\\
		X^*_0=x.
	\end{cases}
	\end{align}
	Furthermore, the feedback investment and dividend policies $(\pi^*_t,c^*_t):=\big(\pi^*(X^*_t, z^*_t), c^*(X^*_t, z^*_t)\big)$ are admissible.
	
	\item[($iv$)] $V(x, z) = \E^x \displaystyle \int_0^{\tau} \ee^{-\del t} \, \frac{(c_t^*)^{1-p}}{1-p} \, \dd t$.
\end{enumerate}

	We prove these conditions below.\medskip

\noindent	\textbf{Condition $(iii)$:} By construction of $\pis$ and $\cs$, if the processes $(X^*_t)$ and $(z^*_t)$ satisfying \eqref{eq:XSTAR} exist, then $(\pi^*_t,c^*_t):=\big(\pi^*(X^*_t, z^*_t), c^*(X^*_t, z^*_t)\big)$ are admissible. To show this, we exploit the property of the feedback control functions $\pi^*$ and $c^*$ that was explained in Remark \ref{rem:MStar}. Namely, following these feedback controls, we have that the historical peak of the dividend rate $(z^*_t)$ satisfies  $z^*_t = \frac{M^*_t}{\ws}$ for all $t\ge0$, where $(M^*_t)$ is given by
\begin{align}
	M_t^* = \max \Big\{ M_0^*, ~ \max_{0 \le s <  t} X_s \Big\}; \quad t > 0,
\end{align}
and $M^*_0 = \ws z_0$. Therefore, it suffices to
% show that, for all $x,z>0$, the following stochastic differential equation (SDE) has a unique strong solution $(X^*_t)_{t\ge0}$,
% \begin{align}
% \begin{cases}
% 	\dd X_t^* = \Big(\mu\,\pi^*\left(X^*_t, \frac{M^*_t}{\ws}\right) - c^*\left(X^*_t, \frac{M^*_t}{\ws}\right)\Big) \dd t + \sig \pi^*\left(X^*_t, \frac{M^*_t}{\ws}\right) \, \dd W_t;\quad t\ge0,\\
% 	M^*_t = \max \left\{\ws\,z, \displaystyle\sup_{0 \le s < t} X^*_s\right\};\quad t\ge0,\\
% 	X^*_0=x.
% \end{cases}
% \end{align}
% It only remains to
show that \eqref{eq:XSTAR} has a unique strong solution.
This is a path-dependent SDE. By Theorem 7 in Section 3 of Chapter 5 of \cite{Protter2005}, it suffices to show that the functionals
\begin{align}
	\Gv(t, \xv) := \pi^*\Bigg(\xv(t),\frac{1}{\ws}\, \max \left\{\ws\,z, \displaystyle\sup_{0 \le s < t} \xv(s)\right\} \Bigg)
\intertext{and}
	\Fv(t, \xv) := c^*\Bigg(\xv(t), \frac{1}{\ws}\, \max \left\{\ws\,z, \displaystyle\sup_{0 \le s < t} \xv(s)\right\}\Bigg),
\end{align}
defined for $t\ge0$ and for continuous functions $\xv:\Rb_+\to\Rb_+$, are functional Lipschitz in the sense of \cite{Protter2005}. This property follows from the Lipschitz property of the $\pi^*$ and $c^*$ given in Lemma \ref{lem:Lipschitz} below. In particular, for all $t\ge0$ and continuous functions $\xv$ and $\yv$, we have
\begin{align}
	|\Gv(t,\xv) - \Gv(t,\yv)| &\le K\left[ \big| \xv(t) - \yv(t) \big|
	+ \left| \max \left\{\ws\,z, \displaystyle\sup_{0 \le s < t} \xv(s)\right\}
		 - \max \left\{\ws\,z, \displaystyle\sup_{0 \le s < t} \yv(s)\right\}\right| \right]\\*
	&\le 2\,K\, \displaystyle\sup_{0 \le s \le t} \big| \xv(s)-\yv(s) \big|.
\end{align}
Thus, $\Gv$ is functional Lipschitz. That $\Fv$ is functional Lipschitz follows similarly.
\begin{lemma}\label{lem:Lipschitz}
	The functions $\pis(x, z)$ and $\cs(x, z)$ are Lipschitz on $\Rb_+^2$.
\end{lemma}
\begin{proof}
	We prove the statement for $\pis$; the Lipschitz property of $\cs$ follows similarly. By \eqref{eq:pis} and \eqref{eq:pis_xbig}, we have
	\begin{align}\label{eq:Aux1}
		\pis(x,z) = z\, f\bigg(\frac{x}{z}\bigg),
	\end{align}
	in which $f$ is defined by $f(w) = \pis(w,1)$. Note that $f$ is continuously twice differentiable on $[0, \ws]$, and $f(w) = K w$ for $w\ge \ws$, in which $K$ equals
	\begin{align}
		K = - \, \dfrac{\mu}{\sig^2} \, \dfrac{U_w(\ws)}{\ws U_{ww}(\ws)}.
	\end{align}
It follows that $f$ has bounded first and second derivatives. Differentiating \eqref{eq:Aux1} with respect to $x$ yields
	\begin{align}
		\Big| \pis_x(x,z) \Big| = \left| f' \bigg(\frac{x}{z}\bigg) \right| \le \sup_{w\ge0} \Big| f'(w) \Big| < \infty.
	\end{align}
	Furthermore, by differentiating \eqref{eq:Aux1} with respect to $z$, we obtain
	\begin{align}
		\Big| \pis_z(x,z) \Big| = \left| f\bigg(\frac{x}{z}\bigg) - \frac{x}{z}\,f'\bigg(\frac{x}{z}\bigg) \right| 
		\le \frac{(\ws)^2}{2}\sup_{0\le w\le \ws} \Big| f''(w) \Big| < \infty.
	\end{align}
	To obtain the inequality, we used the fact that, for $\frac{x}{z}\ge \ws$, 
	\begin{align}
		f\left(\frac{x}{z}\right) - \frac{x}{z}\,f'\left(\frac{x}{z}\right) = 0,
	\end{align}
	and that, by Taylor's approximation,
	\begin{align}
		\left| f\bigg(\frac{x}{z}\bigg) - \frac{x}{z}\,f'\bigg(\frac{x}{z}\bigg) \right| \le \frac{(\ws)^2}{2}\sup_{0\le w\le \ws} \Big| f''(w) \Big|,
	\end{align}
	for all $0<\frac{x}{z}\le \ws$. Finally, we deduce that $\pis$ is Lipschitz since it has bounded derivatives.
\end{proof}
\medskip

\noindent\textbf{Condition $(iv)$:} Define
	\begin{equation}
		\tauh_n = n \wedge \tau \wedge \inf \left\{t > 0: \int_0^t \ee^{-\del s} (\pi_s^*)^{2} V_x^2(X_s^*,z^*_s) \dd s \ge n\text{ or } z^*_t \ge n \right\},
	\end{equation}
	and note that $\tau_n\to \tau$ a.s.\ because of continuity of $(X^*_t, z^*_t)_{t>0}$.
	% \begin{equation}
	% 	\tau^* = \inf \big\{ t \ge 0 : X^*_t \le 0 \big\}.
	% \end{equation}
	Repeating the argument in the proof of Lemma \ref{lem:verif} for $\pi=\pis$, $c=\cs$, and $\tau_n=\tauh_n$ yields
	\begin{equation}\label{eq:optim-1}
	\begin{split}
		 V(x,z)&= 
		 \E^x \int_0^{\tauh_n} \ee^{-\del t} \, \frac{(c_t^*)^{1-p}}{1-p} \, \dd t
		 + \Eb^x\Big(\ee^{-\del {\tauh_n}} \, V(X^*_{\tauh_n}, z^*_{\tauh_n})\Big)\\
		&\quad- \E^x \int_0^{\tauh_n} \ee^{-\del t} \Bigg[ \frac{1}{2} \sig^2 (\pi^*)^{2}_t V_{xx} + (\mu \pis_t - \cs_t) V_x - \del V + \frac{(c_t^*)^{1-p}}{1-p}\Bigg] \dd t\\
		&\quad- \E^x \bigg(\int_0^{\tauh_n} \ee^{-\del t} \, V_z \, \dd (z_t^*)^{c} + \sum_{0 \le t \le {\tauh_n}} \ee^{-\del t} \Big(V(X^*_t,z^*_{t^+}) - V(X^*_t,z^*_t)
		\Big) \bigg).
	\end{split}
	\end{equation}
	
	As stated in Remark \ref{rem:MStar} and the discussion that precedes it, the process $(z^*_t)$ defined in \eqref{eq:XSTAR} can only have a jump at $t=0$ and only if $x>\ws z$. After this possible initial jump, the process $(X^*_t, z^*_t)_{t > 0}$ will be kept in the domain $\Dc$. Because $V$ satisfies the free-boundary problem in \eqref{eq:HJB-FB} and $(X^*_t, z^*_t)\in\Dc$ for $t>0$, we have
	\begin{align}
		\E^x \int_0^{\tauh_n} \ee^{-\del t} \Bigg[ \frac{1}{2} \sig^2 (\pi^*)^{2}_t V_{xx} + (\mu \pis_t - \cs_t) V_x - \del V + \frac{(c_t^*)^{1-p}}{1-p}\Bigg] \dd t=0.
	\end{align}
	Furthermore, because $V_z(x,z)=0$ on $\Rb_+^2 - \Dc$ and $(z^*_t)$ can only jump at $t=0$, we have
	\begin{align}
		\E^x \bigg(\int_0^{\tauh_n} \ee^{-\del t} \, V_z \, \dd (z_t^*)^{c} + \sum_{0 \le t \le {\tauh_n}} \ee^{-\del t} \Big(V(X^*_t,z^*_{t^+})-V(X^*_t,z^*_t)
				\Big) \bigg)
		= V(x,z^*_{0^+})-V(x,z) = 0.
	\end{align}
	From \eqref{eq:optim-1}, it then follows that
	\begin{align}\label{eq:optim-2}
		 V(x,z)= 
		 \E^x \int_0^{\tauh_n} \ee^{-\del t} \, \frac{(c_t^*)^{1-p}}{1-p} \, \dd t
		 + \Eb^x\Big(\ee^{-\del {\tauh_n}} \, V(X^*_{\tauh_n}, z^*_{\tauh_n})\Big).
	\end{align}
	
	Next, we prove that
	\begin{align}\label{eq:transversality}
		\underset{n\to\infty}{\lim\inf}\, \Eb^x \Big(\ee^{-\del \tauh_n} V(X^*_{\tauh_n}, z^*_{\tauh_n}) \Big) = 0.
	\end{align}
	For $\al=0$, this equation follows from the so-called \emph{transversality condition} of the value function in the classical Merton's problem. By Lemma \ref{lem:V_al} below, $V$ decreases as $\al$ increases. Thus, \eqref{eq:transversality} is also satisfied for all $\al \in (0, 1)$.
	
\begin{lemma}\label{lem:V_al}
The expression for $V$ in \eqref{eq:V-sol} and \eqref{eq:V_xbig} decreases with respect to $\al$, and as $\al$ approaches $0$, $V$ in \eqref{eq:V-sol} and \eqref{eq:V_xbig} approaches
\begin{equation}\label{eq:V_al0}
\dfrac{x^{1-p}}{1 - p} \( \dfrac{\kap p^2}{p(1 + \kap \del) - 1} \)^p,
\end{equation}
for all $x \ge 0$, independent of $z$.
% Moreover, the expression in \eqref{eq:V_al0} satisfies the transversality condition, that is, Condition $(iv)$ in Corollary \ref{cor:Optimal}.
\end{lemma}

\begin{proof}
From Corollary \ref{cor:Uh_al}, we know that $\Uh$ decreases with $\al$, which implies that $U$ decreases with $\al$; thus, the expression in \eqref{eq:V-sol} also decreases with $\al$.  To see that the expression in \eqref{eq:V_xbig} when $x > \ws z$ decreases with $\al$, differentiate to obtain
\begin{equation}
\dfrac{\partial V}{\partial \al} \propto \dfrac{\kap p^2}{p(1 + \kap \del) - 1} - \ws,
\end{equation}
because $\ws$ increases with $\al$, as we show in Corollary \ref{cor:ws_al} below.  This expression is negative for $\al \in (0, 1)$ because $\lim \limits_{\al \to 0+} \ws = \frac{\kap p^2}{p(1 + \kap \del) - 1}$.

Next, if we allow $\al$ to approach $0$, $\ys$ approaches $1$, $\wa$ approaches $0$, and $\wo$ and $\ws$ both approach $\frac{\kap p^2}{p(1 + \kap \del) - 1}$.  By substituting these limits in $V$ in \eqref{eq:V-sol} and \eqref{eq:V_xbig}, we obtain the expression in \eqref{eq:V_al0}.
% which satisfies the transversality Condition $(iv)$ in Corollary \ref{cor:Optimal} because it is the value function of the corresponding Merton problem, which is well known to satisfy the transversality condition.
\end{proof}	
	
Finally, by letting $n\to\infty$ in \eqref{eq:optim-2} and using $V(0,z) = 0$, \eqref{eq:transversality}, and the monotone convergence theorem, we obtain $V(x,z) = \E^x \displaystyle \int_0^{\tau} \ee^{-\del t} \, \frac{(c_t^*)^{1-p}}{1-p} \, \dd t$.

\end{document}